\documentclass[pra, two column, superscriptaddress, showpacs,nofootinbib,longbibliography, title page]{revtex4-2}

\usepackage[T1]{fontenc}
\usepackage{physics}
\usepackage{enumerate}
\usepackage{etex}
\usepackage{amsmath,amssymb,amsthm,amstext, mathrsfs}
\usepackage[pdftex]{graphicx}
\usepackage{verbatim}
\usepackage{bbold}
\usepackage{mathtools}
\usepackage{pgfplots}
\pgfplotsset{width=8 cm,compat=1.8}
\usepackage{times,txfonts}
\usepackage{braket}
\usepackage{color}
\usepackage{array}
\usepackage{blkarray}
\usepackage{mathtools}
\usepackage{latexsym}
\usepackage{tabularx, booktabs}
\usepackage{graphics,epstopdf}
\usepackage{graphicx}
\usepackage{float}
\usepackage{amsfonts}
\usepackage{subfigure}
\usepackage{hyperref}
\usepackage{soul}
\hypersetup{
	colorlinks=true,
	linkcolor=red,
	filecolor=blue,      
	urlcolor=blue,
	citecolor=purple,
}

\DeclareMathOperator*{\Motimes}{\text{\raisebox{0.25ex}{\scalebox{0.65}{$\bigotimes$}}}}

\newtheorem{Corollary}{Corollary}

\newtheorem{definition}{Definition}

\newtheorem{thm}{Theorem}
\newtheorem{Lemma}{Lemma}
\newtheorem{remark}{Remark}

\begin{document}

\title{Scalable Self-Testing of Mutually Anticommuting Observables and Maximally Entangled Two-Qudits}

\author{Souradeep Sasmal}
\email{souradeep.007@gmail.com}
\affiliation{Institute of Fundamental and Frontier Sciences, University of Electronic Science and Technology of China, Chengdu 611731, China}

\author{Ritesh K. Singh}
\email{riteshcsm1@gmail.com}
\affiliation{Department of Physics, Indian Institute of Technology Hyderabad, Kandi, Sangareddy, Telangana 502285, India}

\author{Prabuddha Roy}
        \email{prabuddhar@kias.re.kr}
         \affiliation{Quantum Universe Center, Korea Institute for Advanced Study, Seoul 02455, Republic of Korea }

\author{A. K. Pan}
	\email{akp@phy.iith.ac.in}
	 \affiliation{Department of Physics, Indian Institute of Technology Hyderabad, Kandi, Sangareddy, Telangana 502285, India}


\begin{abstract}
The next frontier in device-independent quantum information lies in the certification of scalable and parallel quantum resources, which underpin advanced quantum technologies. We put forth a simultaneous self-testing framework for maximally entangled two-qudit state of local dimension $m_*=2^{\lfloor n/2 \rfloor}$ (equivalently $\lfloor n/2 \rfloor$ copies of maximally entangled two-qubit pairs), together with $n$ numbers of anti-commuting observables on one side. To this end, we employ an $n$-settings Bell inequality comprising two space-like separated observers, Alice and Bob, having $2^{n-1}$ and $n$ number of measurement settings, respectively. We derive the local ontic bound of this inequality and, crucially, employ the Sum-of-Squares decomposition to determine the optimal quantum bound without presupposing the dimension of the state or observables. We then establish that any physical realisation achieving the maximal quantum violation must, up to local isometries and complex conjugation, correspond to a reference strategy consisting of a maximally entangled state of local dimension of at least $2^{\lfloor n/2 \rfloor}$ and local observables forming an irreducible representation of the Clifford algebra. This construction thereby demonstrates that the minimal dimension compatible with $n$ mutually anticommuting observables is naturally self-tested by the maximal violation of the proposed Bell functional. Finally, we analyse the robustness of the protocol by establishing quantitative bounds relating deviations in the observed Bell value to the fidelity between the realised and the ideal strategies. Our results thus provide a scalable, dimension-independent route for the certification of high-dimensional entanglement and Clifford measurements in a fully device-independent framework.
\end{abstract}

\maketitle


\section{Introduction} 

With the rapid advancement of quantum technology, many promising device-independent (DI) information protocols increasingly require simultaneous access to multiple identically prepared entangled states (or, equivalently, higher dimensional entangled states) in order to scale single-copy schemes towards practical, large-scale implementations. For example, parallel DI key distribution protocols \cite{Cerf2002, Georgios2005, Mafu2013, Vidick2017, Jain2020} exploit multiple maximally entangled two-qubit states to enhance secure key rates. In DI randomness expansion \cite{Pironio2010}, the amount of certified randomness scales with the number of entangled states certified, as demonstrated in unbounded randomness expansion protocols \cite{Agresti2021}. Similarly, entanglement detection for arbitrary mixed states \cite{Bowles2018}, and verifiable delegated or blind quantum computing protocols \cite{Reichardt2013, Gheorghiu2015, Coladangelo2017ar, Jain2020, Gottesman1999, Kapit2016, Hu2020, Kues2017}, require a client to prepare and supply multiple entangled qubit pairs to a server. In all such cases, the ability to certify multiple entangled two-qubit states and their corresponding measurement settings is not merely of theoretical interest but a fundamental step towards building practical, scalable, and secure quantum networks and cloud-based quantum computation platforms.

The most stringent form of certification of quantum states and measurements, henceforth referred to as the `quantum strategy’, is provided by the DI approach. In this framework, experiments are characterised solely by observed input-output statistics, without assuming any knowledge of the underlying quantum devices. Under the assumptions of nosignalling and the validity of quantum theory, it has been shown that the maximal violation of the CHSH inequality \cite{CHSH1970, Bell1964} uniquely identifies maximally entangled two-qubit state \cite{Tsirelson1987, Werner1987, Popescu1992, Tsirelson1993}. Mayers and Yao \cite{Mayers2003} further demonstrated that any two quantum strategies achieving this maximal violation of the CHSH inequality are related by local isometries, thereby coining the term `self-testing', a method of certifying both the underlying entangled state and the associated measurements in a fully DI manner. The robustness of this approach was later established in \cite{McKague2012}. 

Beyond the self-testing of maximally entangled two-qubit states, significant progress has been made towards self-testing of quantum instruments \cite{Mohan2019, Shi2019, Tavakoli2020, Roy2023, Paul2024}, arbitrary two-qubit pure states, both using two-outcome tilted Bell inequalities \cite{Yang2013, Bamps2015, Coladangelo2017, Coopmans2019, Kaniewski2020, Valcarce2022}, mixed entangled states and arbitrary extremal measurements \cite{Sarkar2026}, and through approaches not reliant on Bell-like inequalities \cite{Rabelo2012, Rai2021, Rai2022, Ranendu2025}. Furthermore, while extensions to multipartite scenarios \cite{Wu2014, Li2018, Sarkar2022, Panwar2023, Šupić2023, Sarkar2025, Singh2025} have also been developed the self-testing of general high-dimensional entangled systems remains challenging due to the increasingly intricate structure of the corresponding Hilbert spaces \cite{Sarkar2021, Fu2022, Volcic2024}. 

For bipartite systems of local dimension $2^m$, one promising approach involves the self-testing of multiple copies of lower-dimensional states. Since any maximally entangled state of dimension $2^m$ can be decomposed into multiple copies of maximally entangled two-qubit states \cite{Kraft2018}, one can employ sequential self-testing, in which multiple copies are certified one after another through a sequence of Bell tests \cite{Reichardt2013}. However, this sequential approach relies on the stability and non-adaptivity of the measurement devices over time, rendering it vulnerable to adversarial strategies exploiting temporal dependencies.

A more direct approach is parallel self-testing \cite{Coladangelo2017, Wu2016, McKague2016, Coladangelo2017ps, McKague2017, Supic2021}, where multiple copies of a state, or even distinct entangled states \cite{Supic2021}, are certified simultaneously. These schemes rely on performing individual measurements on each copy held by each party, i.e., a Bell test is performed for each subsystem, and the statistics are collected. Both the collective and individual statistics are then used to self-test the multiple copies of maximally entangled two-qubit states. These schemes involve performing independent Bell tests on each subsystem and using both the individual and collective statistics to self-test several entangled two-qubit states concurrently. Related efforts have also targeted direct self-testing of higher-dimensional maximally entangled states \cite{Sarkar2021, Supic2025}, although such protocols typically require multi-outcome measurements that are experimentally demanding. Moreover, some of these protocols only certify observables on one subsystem (e.g., Bob’s anti-commuting measurements) \cite{Supic2025}.

In this work, we introduce a scalable framework for DI certification of the maximally entangled two-qudit states using only dichotomic measurements. Specifically, we devise an $n$-settings, two-outcome Bell inequality, by deriving the local ontic bound of a multi-setting linear functional introduced in the context of a prepare-and-measure communication game in \cite{Ghorai2018}. We subsequently obtain the optimal quantum bound of this inequality via a sum-of-squares (SOS) decomposition \cite{Bamps2015}, without assuming any prior knowledge of the underlying Hilbert-space dimension. This analysis leads to the self-testing of maximally entangled states of local dimension $m_*=\lfloor n/2\rfloor$ (equivalently, $m_*$ copies of maximally entangled two-qubit states \cite{Kraft2018}), together with the corresponding algebra of $n$ mutually anticommuting observables on one side, up to local isometries and complex conjugation (local transposition).

Importantly, since any two-outcome multi-setting Bell inequality with real coefficients and joint correlators is locally transposition invariant, the optimal violation does not, in general, imply strict local isometric equivalence. This distinction becomes crucial for higher-setting inequalities, such as the Elegant Bell Inequality \cite{Gisin2007}, where the inclusion of complex Pauli operators (e.g., $\sigma_y$) breaks strict unitary equivalence under transposition \cite{Andersson2017}. Consequently, the definition of self-testing, as equivalence under local isometries, is generalised to include local transposition equivalence or complex conjugation \cite{Andersson2017}. Our framework therefore establishes self-testing of the target maximally entangled two-qudit state up to local isometry and complex conjugation, aligning with the broaden definition of self-testing.

Moreover, we prove that the self-testing statement is robust to experimental imperfections. In particular, if the observed value of the Bell functional deviates from the optimal quantum value by $\delta$, we show that, up to local isometries, the extracted state and the implemented measurements are $\order{\sqrt{\delta}}$-close to the ideal maximally entangled target state and the corresponding Clifford observables.

The paper is structured as follows. Section~\ref{prelims} reviews basic notions of Bell scenarios, Bell inequalities, and self-testing. Sections~\ref{clb} and \ref{qbsos} introduces the $n$-setting Bell functional \cite{Ghorai2018} and derives its local and quantum bounds for arbitrary $n$. In Section~\ref{selftesting}, we establish the self-testing relations between the physical and reference experiments. Section~\ref{sec5} analyses robustness to experimental noise. Section~\ref{secSD} summarises our main findings and outlines future directions.


\section{Preliminaries}\label{prelims} 
To set the stage for our results, we briefly outline the Bell scenario, the relevant sets of correlations, and the notion of self-testing.

\subsection{Bell scenario: Local, Quantum and Nosignalling correlations} 

A Bell scenario involves two spatially separated and non-communicating parties, Alice and Bob, who share a common preparation $P\in\mathcal{P}$. In each experimental run, Alice performs one of $m_A$ measurements $A_{x}$ with input $x\in [m_A]$, obtaining an outcome $a \in \{1,2,...,O_A\}$. Similarly, Bob performs one of $m_B$ measurements $B_y$ with $y \in [m_B]$, yielding $b \in \{1,2,...,O_B\}$. The resulting statistics define a behaviour vector
\begin{equation}
\vec{P}:=\qty{p\qty(a,b|x,y,P)}_{a,b,x,y} \in \mathbb{R}^{m_Am_BO_AO_B}
\end{equation}
where $p\qty(a,b|x,y,P)\geq 0$ and $\sum_{a,b}p\qty(a,b|x,y,P) =1$ denotes the joint conditional probability of obtaining the observed outcome-pair $(a,b)$. 

In quantum theory, the preparation corresponds to a shared state $\rho \in \mathscr{L}(\mathcal{H}^A \otimes \mathcal{H}^B)$, and the probabilities are governed by the Born rule,
\begin{equation}
    p(a,b|x,y,\rho) = \Tr[\rho A_{a|x}\otimes B_{b|y}]
\end{equation}
where $A_{a|x},B_{b|y}\geq 0$ and $\sum_a A_{a|x}=\sum_b B_{b|y} =\openone$ are the POVM elements corresponding to the outcomes $a$ and $b$ respectively. 

The set of all quantum behaviours forms a convex set $\mathcal{Q}$, known as the quantum set. By Neumark’s theorem \cite{Peres1990}, every element of $\mathcal{Q}$ can be obtained through projective measurements on a pure bipartite state, possibly in a higher-dimensional Hilbert space.

In the ontological model of quantum theory, a specific preparation of $\rho$ essentially prepares the system into an underlying ontic state $\lambda\in\Lambda$ with a probability distribution $\mu(\lambda): \mu(\lambda)\geq0; \ \sum_{\lambda} \mu(\lambda)=1$, where $\Lambda$ denotes the ontic state space. The ontic state $\lambda$ determines the outcome of a measurement of each party, which is independent of the choice of measurement settings and outcomes of the other party (locality). For a local ontic model \cite{Bell1964}, the conditional joint probability in the ontic state space factorises as follows:
\begin{equation} \label{fact}
	p(a,b|x,y,\lambda) = p(a|x,\lambda) \ p(b|y,\lambda) \ \ \forall \ a,b,x,y,\lambda
\end{equation}  
Now, in order to reproduce the prediction of quantum theory from the ontic model, the following reproducibility condition need to be satisfied
\begin{equation}\label{jp}
	p(a,b|x,y,\rho)=\sum\limits_{\lambda\in\Lambda} \mu(\lambda) p(a|x,\lambda)p(b|y,\lambda)
\end{equation} 
The set of all such local correlations constitutes a polytope, denoted by $\mathcal{L}$.  It is known \cite{Goh2018, Rai2019, Le2023, Barizien2025} that the local set $\mathcal{L}\subset \mathcal{Q} \subset \mathcal{NS}$, where $\mathcal{NS}$ is the nosignalling polytope containing all behaviours consistent with relativistic causality \cite{Masanes2006}. Note that unlike $\mathcal{L}$ and $\mathcal{NS}$, the quantum set $\mathcal{Q}$ is not a polytope.


\subsection{Bell functionals and inequalities}

A Bell functional is a linear map $\mathscr{G} \in  \mathbb{R}^{m_Am_BO_AO_B}$ acting on a behaviour $\vec{P}$ as
\begin{equation}
  G(\vec{P}) = \expval{\mathscr{G},\vec{P}}:=\sum_{a,b,x,y} \mathscr{G}_{a,b,x,y}  p(a,b|x,y)
\end{equation}
where $\mathscr{G}_{a,b,x,y}\in \mathbb{R}$. The maximal attainable value over a set $\mathcal{S}$ comprising $\mathcal{NS}$, $\mathcal{Q}$ and $\mathcal{L}$, is given by
\begin{equation}
    (G)^{opt}_{\mathcal{S}}=\max\limits_{\vec{P} \in \mathcal{S}}G(\vec{P})
\end{equation}
A Bell inequality is then expressed as
\begin{equation}\label{bell}
    G(\vec{P}) \leq (G)^{opt}_{\mathcal{L}} 
\end{equation}
and a behaviour is said to be nonlocal if it violates this bound, i.e. $G(\vec{P}) >(G)^{opt}_{\mathcal{L}}$.


\subsection{Optimal quantum Bell violation, Quantum extremality and self-testing}

From convex geometry, a linear functional defines a hyperplane, and at its optimal quantum value, this hyperplane becomes tangent to the quantum set $\mathcal{Q}$. The corresponding behaviour is therefore an exposed extremal point of $\mathcal{Q}$ \cite{Le2023, Barizien2025}. However, extremality of correlations does not guarantee uniqueness of the underlying Hilbert-space realisation. Notably, in the $2-3-2$ Bell scenario, the maximal violation of the $I_{3322}$ inequality does not uniquely determine the quantum strategy \cite{Pal2010}. Therefore, to establish self-testing, one must show that any two quantum realisations producing the same optimal value are related by a local isometry and complex conjugation. Formally,

{\definition An observed behaviour $\vec{P}_{*} \equiv \{p_{*}(a,b|x,y)\}$ self tests a quantum strategy $\qty{\ket{\psi}, A_{x}, B_{y}}$, up to local isometries and complex conjugations, if for any other quantum realisation $\qty{\ket{\psi'}, A'_{x}, B'_y}$ that reproduces the same correlations,
\begin{equation}
p_{*}(a,b|x,y)=\expval{A'_{a|x}\otimes B'_{b|y}}{\psi'}, \ \ \forall a,b,x,y
\end{equation}
and there exist local isometries $\Phi_A: \mathcal{H}_A \mapsto \mathcal{H}'_A \otimes \mathcal{H}_A^{aux}$ and $\Phi_B: \mathcal{H}_B \mapsto \mathcal{H}'_B \otimes \mathcal{H}_B^{aux}$ such that
\begin{equation}\label{stdefimp}
    \qty(\Phi_A \otimes \Phi_B) \ket{\psi}= \ket{\psi'} \otimes \ket{\text{junk}} \ \ \forall a,b,x,y.
\end{equation}
Moreover, certain self-testable behaviours $\vec{P}_{*}$ saturate the quantum bound of a Bell functional $\mathscr{G}$, i.e., $\mathscr{G}.\vec{P}_{*} = (G)^{opt}_{\mathcal{Q}}$.
}


\section{Results: Classical bound of the considered higher settings Bell functional}\label{clb}

We consider a particular Bell scenario, in which Alice has $m_A=2^{n-1}$ and Bob has $m_B=n$ number of measurements. All the measurements yield only two outcomes $+1$ and $-1$. We evaluate the local ontic bound of the following Bell functional
\begin{equation} \label{bellfuncgen}
    \mathscr{G}_n:=   \sum\limits_{x=1}^{2^{n-1}}\sum\limits_{y=1}^{n} (-1)^{z_y^{x}} A_x B_y
\end{equation}
where the value of the functional is denoted by $G_n=\expval{\mathscr{G}_n}$. The term ${z_{y}^{x}}$ takes values of either $0$ or $1$, determined by a set of bit strings of length $n$ denoted by $l\in \qty{0,1}^{n}$. If we fix the first bit as $0$, there will be $2^{n-1}$ distinct $l$ for all $n$. For each $x\in \qty{1,\cdots,2^{n-1}}$, let $z^x:=\qty(z^x_1,\cdots,z^x_n)$ be the associated bitstring. Then $z_y^{x}$ is the element corresponding to $y$-th bit of the  $x$-th string \cite{Ghorai2018}. 

{\thm Any ontological model satisfying the condition of factorisability of the joint conditional probability must satisfy the following inequality
\begin{equation} \label{nbell}
(G_n)_{\mathcal{L}}=\abs{\expval{\sum\limits_{x=1}^{2^{n-1}} \sum\limits_{y=1}^{n} (-1)^{z_y^{x}} A_x  B_y}_\mathcal{\lambda} }\leq \qty(\left\lfloor \frac{n}{2} \right\rfloor +1) \binom{n}{\left\lfloor \frac{n}{2}\right\rfloor +1}
\end{equation} }
For $n=2,3$, the above inequality given by Eq.~(\ref{nbell}) reduces to the well-known CHSH inequality \cite{CHSH1970} and elegant Bell inequality \cite{Gisin2007} respectively. 

\begin{proof}
The factorisability condition given by Eq.~(\ref{fact}) implies
\begin{equation}
  \expval{A_xB_y}_{\lambda} = \sum\limits_{ab}abp(a|x,\lambda) \ p(b|y,\lambda) = \expval{A_x}_{\lambda}\expval{B_y}_{\lambda}
\end{equation}
Employing the factorisability and defining the quantity $\mathcal{B}_{x}:=\qty(\sum\limits_{y=1}^{n} (-1)^{z_y^{x}}B_y)$, from Eq.~(\ref{nbell}), we obtain: 
\begin{equation}
\abs{\expval{\mathscr{G}_{n}}_{L}}=\abs{\sum\limits_{x=1}^{2^{n-1}}\expval{A_x }_{\lambda} \expval{\mathcal{B}_{x}}_{\lambda} } \leq \sum\limits_{x=1}^{2^{n-1}} \abs{ \expval{A_x}_{\lambda} \expval{\mathcal{B}_{x}}_{\lambda}}   \leq \sum\limits_{x=1}^{2^{n-1}} \abs{ \expval{\mathcal{B}_{x}}_{\lambda}}
\end{equation}
First inequality is obtained from the triangle inequality and second inequality is obtained since $\abs{\expval{A_x}_{\lambda}} \leq 1$. As the upper bound in an ontological model is always exhaust by the deterministic bound, without loss of generality we can always take $\expval{B}_y \in \{+1,-1\}$. Therefore, the local ontic bound becomes $(G_{n})_{\mathcal{L}} = \sum\limits_{x=1}^{2^{n-1}} \abs{ \sum\limits_{y=1}^{n} (-1)^{z_y^{x}}}$. Now, by suitably evaluating the quantity $(-1)^{z_y^{x}}$, we determine the local ontic bound $(G_n)_{\mathcal{L}}$ (see Appx.~\ref{thm1}). 
\end{proof}


\section{Results: Optimal quantum bound of $\mathscr{G}_n$}\label{qbsos}
In this section, we evaluate the optimal quantum bound of this $n$-settings Bell inequality and derives the necessary self-testing statements.

{\thm \label{thm2}The optimal quantum value of the $n$-settings Bell functional defined in Eq.~(\ref{bellfuncgen}) is
\begin{equation}\label{opbqf}
    (G_{n})_{\mathcal{Q}}^{opt}=2^{n-1}\sqrt{n}
\end{equation}
This value is attained if and only if the shared state and local observables satisfy the following structural properties.

The shared bipartite state must have full support on the joint kernel of the SOS operators associated with the Bell operator. Consequently, on the support of the state the local observables behave as Hermitian involutions (i.e. Hermitian operators whose square equals the identity, $O^2=\openone$, implying eigenvalues $\pm 1$) obeying pairwise anticommutation relations in expectation, viz.
\begin{equation} \label{antrelsup}
    \begin{aligned}
 \ev**{\qty{B_y,B_{y'}} }{\psi} &= 0 \ \forall y \neq y', \\
 \ev**{\qty{A_x,A_{x'}} }{\psi} &= \frac{2}{n} \sum_{y} (-1)^{z_y^x+z_{y}^{x'}} \ \forall x\neq x',
    \end{aligned}
\end{equation}
These relations imply that the support of the state decomposes as a direct sum of maximally entangled states supported on orthogonal local subspaces of dimension at least $m_k \geq 2^{\lceil \frac{1}{2}(n-1) \rceil}$. Writing the Schmidt decomposition in block form (where a Schmidt block denotes the pair of corresponding orthogonal subspaces $\mathcal{H}_A^{(k)} \subseteq \mathcal{H}_A$ and $\mathcal{H}_B^{(k)} \subseteq \mathcal{H}_B$ associated with a fixed Schmidt coefficient), the state can be expressed as
\begin{equation}\label{blockstatem}
    \ket{\psi}=\bigoplus_k \sqrt{\lambda^{(k)}m_k} \ket{\phi^{+}_{m_k}}, 
\end{equation}
where $\ket{\phi^{+}_{m_k}}$ denotes the maximally entangled state on $\mathcal{H}_A^{(k)} \otimes \mathcal{H}_B^{(k)}$ of local dimension $m_k$, and the coefficients $\lambda^{(k)}>0$ satisfy $\sum_k \lambda^{(k)}=1$.

Within each Schmidt block the observables preserve the block structure and form a representation of the complex Clifford algebra $\mathrm{Cl}_n\qty(\mathbb{C})$. In particular, the global observables decompose as
\begin{equation}
   A_x =\bigoplus_k A_x^{(k)}, \ B_y =\bigoplus_k B_y^{(k)}=\bigoplus_k \qty(\mathcal{A}_y^{(k)})^T,
\end{equation}
with
\begin{equation} \label{aliceobsy}
    \mathcal{A}_y^{(k)}=\frac{\sqrt{n}}{2^{n-1}} \sum\limits_{x=1}^{2^{n-1}} (-1)^{z_y^x} A_x^{(k)} \ \ \forall y \in [n].
\end{equation}
and the operators within each block satisfy
\begin{equation}\label{blockantic}
 \qty{B_y^{(k)},B_{y'}^{(k)}} = \qty{\mathcal{A}_y^{(k)},\mathcal{A}_{y'}^{(k)}} =0 \ \ \forall y\neq y'.
\end{equation}
Any degeneracy in the Schmidt coefficients corresponds solely to classical mixing between identical Clifford blocks and therefore does not alter the optimal quantum value.}

\begin{proof}

The optimal quantum value $(G_n)_\mathcal{Q}^{opt}$ of the Bell functional $\mathscr{G}_n$ defined in Eq.~(\ref{bellfuncgen}) is given by the optimisation problem
\begin{equation}
    ({G}_n)_\mathcal{Q}^{opt}= \sup_{\rho,A_{x},B_y}\ \Tr[\mathscr{G}_{n} \ \rho]
    \end{equation}
where the supremum is taken over all bipartite quantum states $\rho$ and local measurement operators $\qty{A_x}$ and $\qty{B_y}$. While the set of $\rho$ is convex, the joint optimisation over $\rho$ and measurement operators is intrinsically non-convex due to the bilinear dependence of the objective function $\Tr[\mathscr{G}_{n}\ \rho]$ on these variables \cite{Klep2024}. Consequently, direct optimisation is generally intractable. In practice, one may employ semidefinite programming (SDP) relaxations such as the NPA hierarchy to obtain upper bounds on $(G_n)_\mathcal{Q}^{opt}$ \cite{Navascués2007,Navascués2008}.

Here, we adopt the SOS relaxation approach, which reformulates the optimisation as a convex problem by expressing the shifted Bell operator as a sum of positive semidefinite terms \cite{Bamps2015, Ghorai2018},
\begin{equation} \label{starti}
    G_n \openone - \mathscr{G}_n = \sum_{x = 1}^{2^{n-1}} \frac{\omega_x}{2} \ M_x^{\dagger} M_x,
\end{equation}
where $\omega_{x}>0$ are real coefficients and $M_x$ are Hermitian polynomials in $\qty{A_x,B_y}$ acting on the composite Hilbert space $\mathcal{H}^d \otimes \mathcal{H}^d$. Since each $M^{\dagger}_x M_x\geq 0$, it follows that for any $\rho$,
\begin{equation} \label{opti}
    \Tr[\mathscr{G}_n \ \rho] \leq G_n
\end{equation}
Hence, $G_n$ provides an upper bound on the optimal quantum value, and the tightest bound corresponds to the minimal $G_n$ consistent with Eq.~(\ref{starti}),
\begin{equation} \label{sdpprog}
    \min_{M_x} \ G_n \quad \text{s.t.} \quad  (G_n)^{opt}_{\mathcal{Q}} \openone - \mathscr{G}_n = \sum_{x = 1}^{2^{n-1}} \frac{\omega_x}{2} \ M^{\dagger}_x M_x 
\end{equation}
We construct the explicit SOS decomposition by defining
\begin{equation}\label{omega}  
M_{x} = \frac{\openone\otimes \mathcal{B}_x}{\omega_{x}} - A_x \otimes \openone ; \ \ \omega_{x} =||\   \mathcal{B}_x  \ ||_{\rho}\geq 0, 
\end{equation}
where 
\begin{equation} \label{curlyBx}
    \mathcal{B}_x :=\sum\limits_{y=1}^{n} (-1)^{z^x_y} B_y \ \ \ \forall x
\end{equation}
are unnormalised Hermitian operators representing linear combinations of Bob’s observables, and $\norm{\cdot}_{\rho}$ denotes the state-weighted norm, $\norm{\mathcal{B}_x}_{\rho}=\sqrt{\Tr[\mathcal{B}_x^{\dagger}\mathcal{B}_x \ \rho]}$. Assuming dichotomic observables ($A_x^2=B_y^2=\openone$), straightforward algebra yields 
\begin{eqnarray}\label{betasos11}
  \sum_{x = 1}^{2^{n-1}} \frac{\omega_x}{2} \ M_x^{\dagger} M_x=  \qty(\sum\limits_{x=1}^{2^{n-1}}\omega_x)\openone - \mathscr{G}_n.
\end{eqnarray}
A detailed evaluation is provided in Appx.~\ref{detailedSOS}. Comparing coefficients of $\openone$ and $\mathscr{G}_n$ in Eq.~(\ref{sdpprog}) gives
\begin{equation}\label{betasos1}
	(G_n)_{\mathcal{Q}}^{opt}= \max\sum\limits_{x=1}^{2^{n-1}} \sqrt{\Tr[\mathcal{B}_x^{\dagger}\mathcal{B}_x \rho]}=2^{n-1}\sqrt{n}.
\end{equation}
where the final equality follows from the explicit optimisation presented in Appx.~\ref{detailedSOS}.

\paragraph*{Conditions for saturation: } Taking the expectation value of the SOS decomposition (Eq.~\ref{betasos11}) with the optimal state $\rho$, we find
\begin{equation}
\forall x, \   \sum_{x = 1}^{2^{n-1}} \Tr[ M^{\dagger}_x M_x \rho]=0
\end{equation}
which, since each term is non-negative, implies
\begin{equation}\label{kerneleq}
\forall x, \ M_x\ket{\psi}=0  \implies \qty(\openone \otimes \mathcal{B}_x) \ket{\psi}=\sqrt{n}  \qty(A_x\otimes\openone) \ket{\psi}.
\end{equation}
for any purification $\ket{\psi}$ of $\rho$. Each $M_x$ may therefore be regarded as a residual operator that vanishes on the support of the optimal state. In Appx.~\ref{optimalState}, we show that these relations directly imply the anticommutation conditions among the local observables given in Eq.~(\ref{antrelsup}).

\paragraph*{State and observable structure in Schmidt form:} Expressing the optimal state $\ket{\psi}$ in a fixed product basis,
\begin{equation} \label{statefixpro}
    \ket{\psi}=\sum_{i,j=1}^d \mu_{ij} \ket{i}_A \otimes \ket{j}_B=\mathrm{vec}\qty(\mathscr{M})
\end{equation}
where the vectorisation map $\mathrm{vec}: \mathbb{C}^{d \times d} \to \mathbb{C}^{d^2}$ is defined by $\mathrm{vec}\qty(\mathscr{M}):=\sum_{i,j=1}^d \mu_{ij} \ket{i}_A \otimes \ket{j}_B$, and performing the singular value decomposition $\mathscr{M}=U_A \mathcal{D} V_B^{\dagger}$, with $D=\text{diag}\qty(\lambda_1,\ldots,\lambda_r)$ containing the Schmidt coefficients $\sqrt{\lambda_i} > 0$ and $\sum_i \lambda_i =1$, the state becomes
\begin{equation}\label{statefixpro2}
    \ket{\psi}=\sqrt{r}\qty(U_A \otimes V^{\dagger}_B) \qty(\openone \otimes \mathcal{D}^{\frac{1}{2}}) \ket{\phi^{+}}.
\end{equation}
with $r=\mathrm{rank}\qty(\mathscr{M})$. Substituting into the optimality condition (Eq.~\ref{kerneleq}) and using the vectorisation identity $\qty(A \otimes B) \mathrm{vec}(X)=\mathrm{vec}\qty(AXB^T)$, we obtain
\begin{equation}\label{mainoprel}
  \qty(U_A^{\dagger}\mathcal{A}_y U_A) \mathcal{D}^{\frac{1}{2}} = \mathcal{D}^{\frac{1}{2}} \qty(V_B^{\dagger} B_y^T V_B)
\end{equation}
Redefining the unitarily transformed local observables $\mathcal{A}_y \equiv U_A \mathcal{A}_y U_A^{\dagger}$ and $B_y\equiv V_B B_y V_B^{\dagger}$ with corresponding Schmidt bases $\ket{u_i}:=U_A \ket{i}_A$ and $\ket{v_i}:=V_B \ket{i}_B$ with $i \in \qty{1,2,\ldots,r}$, the operator condition simplifies to $\mathcal{A}_y \mathcal{D}^{\frac{1}{2}}=\mathcal{D}^{\frac{1}{2}}B_y^T \ \forall y$.

In the Schmidt basis, $\mathcal{D}$ can be grouped according to its distinct eigenvalues $\lambda^{(k)}$ with multiplicity $m_k$ for $k=1,\ldots,\alpha$, such that $\sum_{i=1}^{\alpha} m_{k}=r$.  Thus, $\mathcal{D}=\oplus_{k=1}^{\alpha} \lambda^{(k)} \openone_{m_k}$, and state decomposes as
\begin{equation}\label{statescmidtgl}
    \ket{\psi}= \bigoplus\limits_{k=1}^{\alpha} \sqrt{\lambda^{(k)} m_k} \ket{\phi^{+}_{m_k}},
\end{equation}
where each $\ket{\phi^{+}_{m_k}} \in \mathcal{H}^{m_k}_A \otimes \mathcal{H}^{m_k}_B$ is maximally entangled state on its degeneracy block.

From Eq.~(\ref{mainoprel}) it follows that both $\mathcal{A}_y$ and $B_y$ commute with $\mathcal{D}$ and are therefore block-diagonal in the same Schmidt decomposition
\begin{equation} \label{obsscmidtgl}
 \forall y, \  \mathcal{A}_y = \bigoplus\limits_{k=1}^{\alpha} \mathcal{A}_y^{(k)}, \ \ \ \text{and }  B_y = \bigoplus\limits_{k=1}^{\alpha} B_y^{(k)}, \ \ \ \mathcal{A}_y^{(k)}, B_y^{(k)} \in \mathscr{L}\qty(\mathcal{H}^{(k)}).
\end{equation}
Within each block, the relation 
\begin{equation}\label{fmainoprel}
    \mathcal{A}_y \mathcal{D}^{\frac{1}{2}}=\mathcal{D}^{\frac{1}{2}}B_y^T \implies \mathcal{A}_y^{(k)} = \qty(B_y^{(k)})^T, \ \forall y,k,
\end{equation}
establishing that Alice’s and Bob’s local observables coincide on each degenerate Schmidt subspace and satisfy identical algebraic relations.

If all Schmidt coefficients are distinct ($\alpha=r$, $m_k=1$), the observables are diagonal and commute, precluding Bell inequality violation. Conversely, if all coefficients are equal ($D=\frac{1}{r}\openone$), the shared state is maximally entangled, and the optimality conditions enforce that the observables form a family of pairwise anticommuting Hermitian involutions satisfying Eq.~(\ref{blockantic}). These operators thus generate a representation of the complex Clifford algebra $\mathrm{Cl}_n\qty(\mathbb{C})$, whose minimal complex representation has local dimension $m_*=2^{\lfloor n/2 \rfloor}$. Accordingly, each degenerate block must have dimension at least $m_k \geq m_*$.

Combining these results, any optimal quantum realisation saturating $\qty(G_n)^{opt}_Q=2^{n-1}\sqrt{n}$ must take the form of Eqs.~(\ref{blockstatem}-\ref{blockantic}). Each block realises an irreducible complex Clifford algebra representation, unique up to local unitaries, while any non-trivial degeneracy corresponds only to classical mixing between identical blocks.
\end{proof}
\begin{remark}
Note that the transpose appearing in the relation $\mathcal{A}_y \mathcal{D}^{\frac{1}{2}}=\mathcal{D}^{\frac{1}{2}}B_y^T$ in Eq.~(\ref{fmainoprel}) arises solely from the adopted vectorisation convention. In the Schmidt basis, where the maximally entangled state is real and symmetric, transposition acts trivially, i.e., $\qty(\mathcal{A}_y \otimes B^T_y) \ket{\phi^+}=\qty(B_y \otimes \mathcal{A}_y^T) \ket{\phi^+}$. Thus, the two quantum realisations $\qty{A_x,B_y^T,\ket{\psi}}$ and its transposed (or complex-conjugated) counterpart $\qty{A_x^T,B_y,\ket{\psi}}$ attains the same maximal Bell violation and are therefore physically indistinguishable with respect to all measurable Bell correlations, even though they are not, in general, related by a local unitary transformation. To distinguish between such conjugate strategies, one must impose additional constraints on the observed statistics beyond standard Bell correlations, as discussed in Ref.~\cite{Andersson2017} in the context of the self-testing properties of Gisin’s elegant Bell inequality~\cite{Gisin2007}.
\end{remark}
These considerations underpin the self-testing analysis developed in Sec.~\ref{selftesting}, where we show that attaining the maximal quantum violation uniquely identifies, up to local unitaries and complex conjugation, the underlying state and observables realising the canonical Clifford-algebraic structure.


\section{Results: Device-Independent Characterisation of the Optimal Quantum Strategy - Self-Testing from the Maximal Violation of $\mathscr{G}_n$} \label{selftesting}

Having established the optimal quantum bound and characterised all realisations that saturate it, we now turn to the self-testing properties of the Bell functional $\mathscr{G}_n$.

Self-testing provides a DI certification of both the shared quantum state and the local observables solely from the observed violation, without assuming any knowledge of the internal workings or the underlying Hilbert space dimension of the quantum system. Consider an arbitrary physical realisation
\begin{equation}\label{idstr}
    \qty{A_{x}\in\mathcal{L}\qty(\mathcal{H}_{A}), B_{y}\in\mathcal{L}\qty(\mathcal{H}_B), \ket{\psi}\in\mathcal{H}_{A}\otimes \mathcal{H}_{B}},
\end{equation}
which achieves the optimal quantum value of the Bell functional $\mathscr{G}_n $ in the ideal, noiseless scenario. The aim of self-testing is to show that any such physical realisation, reproducing the same correlations as the optimal quantum strategy must, up to local isometries and irrelevant degrees of freedom, be equivalent to a fixed reference experiment of fixed local dimension (up to global complex conjugation), denoted by
\begin{equation}\label{refstr}
   \qty{A'_{x}\in\mathcal{L}\qty(\mathcal{H}_{A'}), B'_{y}\in\mathcal{L}\qty(\mathcal{H}_B'), \ket{\psi'}\in\mathcal{H}_{A'}\otimes \mathcal{H}_{B'}},
\end{equation}
The operator relations derived for the optimal violation, particularly the anticommutation conditions and the algebraic structure of the local measurements, enable one to reconstruct the exact form of the target quantum strategy up to local isometries and complex conjugation.

The following analysis formalises this claim and establishes explicit self-testing relations for the optimal quantum strategy associated with $\mathscr{G}_n$.

\begin{thm}[Blockwise self-testing of the Clifford realisation] \label{bscr}
If a quantum strategy $\qty{A_{x}^{(k)}, B_{y}^{(k)}, \ket{\psi^{(k)}}}$ achieves the optimal quantum value of the Bell functional $\mathscr{G}_n$ within the $k^{th}$ Schmidt block, satisfying
\begin{equation}\label{gloobsprop}
\qty{\mathcal{A}_{y}^{(k)},\mathcal{A}_{y'}^{(k)}}=\qty{B_{y}^{(k)},B_{y'}^{(k)}}=2 \delta_{yy'}, \ \ \mathcal{A}_{y}^{(k)}=\qty(B_y^{(k)})^T,
\end{equation}
then there exist local isometries
\begin{equation} \label{hspliiso}
    \Phi_A^{(k)}:\mathcal{H}_A^{(k)}\to \mathcal{H}_{A'}^{(k)} \otimes \mathcal{H}_{J_A^{(k)}}, \ \ \Phi_B^{(k)}:\mathcal{H}_B^{(k)}\to \mathcal{H}_{B'}^{(k)} \otimes \mathcal{H}_{J_B^{(k)}},
\end{equation}
implemented by local unitaries $U_A^{(k)}$ and $V_B^{(k)}$, such that
\begin{equation} \label{obscliff}
 \Phi_B^{(k)}B_{y}^{(k)}\qty(\Phi_B^{(k)})^{\dagger}=\Gamma_y \otimes \openone_{J_B^{(k)}}, \ \ \Phi_A^{(k)}\mathcal{A}_{y}^{(k)}\qty(\Phi_A^{(k)})^{\dagger}=\Gamma_y^T \otimes \openone_{J_A^{(k)}},
\end{equation}
for all $y\in[n]$ and the isometry acts on the optimal physical state as
\begin{equation} \label{statecliff}
    \qty(\Phi_A^{(k)}\otimes \Phi_B^{(k)}) \ket{\psi^{(k)}} = \ket{\phi^+_{m_{*}}}_{A'B'}\otimes \ket{junk}_{J_A^{(k)}J_B^{(k)}},
\end{equation}
where  $\ket{junk}_{J_A^{(k)}J_B^{(k)}}$ is an auxiliary (uncorrelated) state. The operators $\qty{\Gamma_y}$ denote the canonical Clifford generators, acting irreducibly on a minimal subspace of dimension $m_{*}=2^{\lfloor n/2 \rfloor}$. The subsystems $\mathcal{H}_{J_A^{(k)}}$ and $\mathcal{H}_{J_B^{(k)}}$ represent ancillary `junk' spaces that are uncorrelated and do not contribute to the observed correlations.
\end{thm}
\begin{proof}
For clarity, we work with the Alice's effective observables $\mathcal{A}_y$, as defined in Eq.~(\ref{aliceobsy}). Within each Schmidt block $k$, the physical experiment is equivalent to the reference experiment in the sense that there exists a local unitary transformation which decomposes the Hilbert space into a known reference subspace and an auxiliary `junk' subspace, as expressed in Eq.~(\ref{hspliiso}). 

The existence of such a local unitary immediately implies the existence of a corresponding local isometry capable of extracting the known reference strategy by discarding (or tracing out) the junk subsystem. One of the standard constructive ways to realise this is through the SWAP-isometry method \cite{McKague2012, Yang2013, Supic2020}, in which an isometry acts jointly on the physical system and a suitably prepared ancillary system. The isometry effectively `swaps' the properties from the unknown physical subsystem into a well-defined ancillary register of known dimension, thereby reproducing the ideal reference strategy while isolating irrelevant degrees of freedom.

Operationally, each isometry can be implemented by a local unitary acting on an extended system as $\Phi_A^{(k)}\qty(\cdot)=U^{(k)}_A \qty(\cdot \otimes \ket{0}_{J_A^{(k)}})$, and $\Phi_B^{(k)}\qty(\cdot)=U^{(k)}_B \qty(\cdot \otimes \ket{0}_{J_B^{(k)}})$. We now demonstrate the existence of local unitaries $U_A^{(k)}$ and $V_B^{(k)}$ implementing the isometries of Eq.~(\ref{hspliiso}). The construction proceeds by iteratively block-diagonalising the observables to obtain the canonical Clifford representation. To show the iterative relation we are relabelling the observable as $B_y^{(k)}$ to $B_{y,n}^{(k)}$, and similarly for Alice. We first diagonalise $B_{1,n}^{(k)}$ using a suitable local unitary $V_1^{(k)} \in \mathcal{L}\qty(\mathcal{H}_B^{(k)})$ such that
\begin{equation}\label{b1m}
    \mathtt{B}_{1,n}^{(k)} = V_1^{(k)}B_{1,n}^{(k)}\qty(V_1^{(k)})^\dag=\sigma_z \otimes \openone_{m_*/2}\otimes\openone_{J_B}
\end{equation}
where $m_{*}=\lceil \frac{1}{2}\qty(n-1)\rceil$ and $\openone_{m_*/2}$ is a $\qty(m_*/2)\times\qty(m_*/2)$ identity matrix. As $\{B_{1,n}^{(k)},B_{s,n}^{(k)}\}=0 \implies \{\mathtt{B}_{1,n}^{(k)},\mathtt{B}_{s,n}^{(k)}\}=0 \ \forall s\in\{2,3,4,\dots,n\}$. As a consequence, in the same basis where $\mathtt{B}_{1,n}^{(k)} $ is diagonalised via $V_1$, the remaining observables becomes off-diagonal and can be defined as $\mathtt{B}_{s,n}^{(k)} = V_1^{(k)}B_{s,n}^{(k)}(V_1^{(k)})^\dag $ such that $\mathtt{B}_{s,n}^{(k)}$ are of the form
\begin{equation}
    \mathtt{B}_{s,n}^{(k)} = \begin{bmatrix}
        0 & X_{s,n}\\
        X_{s,n}^\dag & 0
    \end{bmatrix} ; \ \forall s\geq2
\end{equation}
We next define a unitary
\begin{equation}\label{Un}
    V_{2}^{(k)} = \begin{bmatrix}
        \openone_{m_*/2} & 0\\
        0 & -\iota X_{2,n}
    \end{bmatrix} \in \mathscr{L}\qty(\mathcal{H}_{B'}^{(k)}),
\end{equation}
which maps
\begin{equation}\label{b2m}
  V_{2}^{(k)}\mathtt{B}_{2,n}^{(k)} \qty(V_{2}^{(k)})^\dag = \qty(-\sigma_y\otimes  \openone_{m_*/2})\otimes\openone_{J_B}
\end{equation}
while preserving the anticommutation with $\mathtt{B}_{1,n}^{(k)}$. Now, using $\qty{\mathtt{B}_{2,n}^{(k)},\mathtt{B}_{s,n}^{(k)}} = 0$ for all $s\in\{3,4,\dots,n\}$, the remaining observables then take the recursive form
\begin{equation}\label{simp}
\begin{aligned}
   V_{2}^{(k)}\mathtt{B}_{s,n}^{(k)}\qty(V_{2}^{(k)})^\dag &= \begin{bmatrix}
        0 & \iota X_{s,n}X_{2,n}^\dag\\
        \iota X_{s,n} X_{2,n}^\dag & 0
    \end{bmatrix} = \sigma_x \otimes \iota X_{s,n} X_{2,n}^\dag
    \end{aligned}
\end{equation}
We have already obtained explicit forms of $\mathtt{B}_{1,n}^{(k)}$ and $\mathtt{B}_{2,n}^{(k)}$ from Eqs.~(\ref{b1m}) and (\ref{b2m}) in the known dimension $m_*$. We now seek the explicit representation of the remaining $(n-2)$ mutually anticommuting observables. To this end, we begin by defining $\mathscr{B}_{s,n} = \iota X_{s,n} X_{2,n}^\dag \ \forall s\in \{3,4,\dots,n\}$. By applying the same arguments used at the beginning of the proof, we obtain $\mathscr{B}_{3,n} = \sigma_z\otimes\openone_{m_*/4}\otimes\openone_{J_B}$ and $\mathscr{B}_{4,n} = -\sigma_y\otimes\openone_{m_*/4}\otimes\openone_{J_B}$. We repeat this procedure up to $s=n$, following the details in \cite{Singh2025njp}, from which we obtain the required recurrence relation.
\begin{eqnarray}\label{byp}
    V_{B}^{(k)}B_{y,n}^{(k)}\qty(V_{B}^{(k)})^\dag=
    \begin{cases}
    B_{y,n}^\prime  \otimes \openone_{J_B} & \forall y\in \{1,2\} \vspace{0.3 cm}\\
\sigma_{x} \otimes B^\prime_{{y-2},{n-2}}  \otimes \openone_{J_B}& \forall y \in \{3,4,\dots,n\} 
    \end{cases}
\end{eqnarray}
where $ B^\prime_{1,k} = \sigma_z\otimes\openone_{2^{l-1}}$ and $B^\prime_{2,k} = -\sigma_y\otimes\openone_{2^{l-1}}$ with $l=\lceil \frac{1}{2} \qty(s-1)\rceil \ \forall s \in [n]$.

An identical recursive argument holds for Alice’s observables $\mathcal{A}_{y,n}^{(k)}$ using unitaries $U_A^{(k)}\in \mathscr{L}\qty(\mathcal{H}_A^{(k)})$. We start with the observable $\mathcal{A}^{(k)}_{1,n}$, which we will diagonalise using the unitary $U_1^{(k)}$.Using the anticommutativity relation $\{\mathcal{A}^{(k)}_{1,n},\mathcal{A}^{(k)}_{s,n}\}=0 \ \forall s\geq 2$, we get
\begin{equation}
    U_1^{(k)}\mathcal{A}_{s,n}^{(k)}\qty(U_1^{(k)})^\dag = \begin{bmatrix}
        0 & Y_{s,n}\\
        Y^\dag_{s,n} & 0
    \end{bmatrix} \ \ \forall s\geq 2
\end{equation}
Next, we construct a unitary similar to Eq.~(\ref{Un}) of the form 
\begin{equation}
    U_{2}^{(k)}=\begin{bmatrix}
        \openone_{m_*/2} & 0\\
        0 &\iota Y_{2,n}
    \end{bmatrix}
\end{equation}
And for the rest of the observables, the repetitive recursion steps will lead to  
\begin{eqnarray}\label{ayp}
    U_{A}^{(k)}\mathcal{A}_{y,n}^{(k)}\qty(U_{A}^{(k)})^\dag=
    \begin{cases}
   \mathcal{A}_{y,n}^\prime \otimes \openone_{J_A}   & \forall y\in \{1,2\} \vspace{0.3 cm}\\
\sigma_{x} \otimes \mathcal{A}^\prime_{{y-2},{n-2}}\otimes \openone_{J_A} & \forall y \in \{3,4,\dots,n\} 
    \end{cases}
\end{eqnarray}
where $\mathcal{A}^\prime_{1,k} = \sigma_z\otimes\openone_{2^{l-1}}$ and $\mathcal{A}^\prime_{2,k} = \sigma_y\otimes\openone_{2^{l-1}} $ with $l=\lceil \frac{1}{2} \qty(s-1)\rceil \  \forall s \in [n]$. Thus the pair $\qty(U_A^{(k)},V_B^{(k)})$ maps the observables to their canonical Clifford form $\qty{\Gamma_y}$, acting on the minimal subspace of dimension $m_{*}=2^{\lfloor n/2 \rfloor}$, while leaving ancillary subsystems untouched, as given in Eq.~(\ref{obscliff}). 

Under this unitary transformation $\mathcal{U}^{(k)}=U_A^{(k)}\otimes V_B^{(k)}$, the observables take the canonical Clifford form
\begin{equation} \label{canoobsblock}
    \mathcal{U}^{(k)}\qty(\mathcal{A}_y^{(k)} \otimes B_y^{(k)})\qty(\mathcal{U}^{(k)})^\dagger =\Gamma_y^T \otimes \Gamma_y,
\end{equation}
and the block state transformation $\mathcal{U}^{(k)} \ket{\psi^{(k)}}$ therefore satisfies the following relations
\begin{equation}
    \qty(\Gamma_y^T \otimes \Gamma_y) \mathcal{U}^{(k)} \ket{\psi^{(k)}} = \mathcal{U}^{(k)} \ket{\psi^{(k)}}, \ \ \forall y \in [n].
\end{equation}
By expanding the transformed state in the canonical computational basis
\begin{widetext}
    \begin{equation}
    \mathcal{U}^{(k)}\ket{\psi^{(k)}}=\sum_{i_1i_2\ldots i_n,j_1,j_2\ldots j_n\in \{0,1\}}\ket{i_{1}i_{2}\ldots i_{n}j_{1}j_{2}\ldots j_{n}}_{A^\prime B^\prime}\otimes \ket{\psi^{(k)}_{i_{1}i_{2}...i_{n}j_{1}j_{2}...j_{n}}}_{{J_A^{(k)}} {J_B^{(k)}}}
\end{equation}
\end{widetext}
where, the bit strings $(i_1i_2\ldots i_n)$ and $(j_1,j_2\ldots j_n)$ label the basis of the reference subsystems $\mathcal{H}_{A'}^{(k)}$ and $\mathcal{H}_{B'}^{(k)}$. Define recursively, as before,
\begin{equation}\label{pside}
  \ket{\psi^{(k)}_{s,n}(i_1,..,i_{s-1},j_1,..,j_{s-1})} = \sum_{i_r, j_r,r\geq s} \qty(\Motimes_{r=s}^n\ket{i_r j_r})\otimes\ket{\psi^{(k)}_{i_{1}..i_{n}j_{1}..j_{n}}}
\end{equation}
so that for $m=1$, $\ket{\psi^{(k)}_{1,n}(i_1,j_1)} =\mathcal{U}^{(k)}\ket{\psi^{(k)}}$. Hence
\begin{equation}\label{m1}    \mathcal{U}^{(k)}\ket{\psi^{(k)}}=\sum_{i_{1},j_{1}}\ket{i_{1}j_{1}}\otimes \ket{\psi^{(k)}_{2,n}(i_{1},j_{1})}.
\end{equation}
Applying the stabiliser for $y=1$, for the first Clifford generator,
\begin{equation}
    \Gamma_1^T \otimes \Gamma_1 = \qty(\sigma_z \otimes \sigma_z) \otimes \openone_{m_*/2} \otimes \openone_{m_*/2}
\end{equation}
the stabiliser condition gives
\begin{equation} \label{y1relatmea}
\begin{aligned}
       \qty(\Gamma_1^T \otimes \Gamma_1 ) \sum_{i_1,j_1}\ket{i_1j_1}\ket{\psi^{(k)}_{2,n}(i_{1},j_{1})}
        = \sum_{i_1,j_1}\ket{i_1j_1}\ket{\psi^{(k)}_{2,n}(i_{1},j_{1})}
\end{aligned}
\end{equation}
Now, lets evaluate the action of $\sigma_z\otimes\sigma_z$. Since in the computational basis $\sigma_z\ket{0}=+\ket{0}$ and $\sigma_z\ket{1}=-\ket{1}$, 
\begin{equation}
    \sigma_z\otimes\sigma_z\ket{i_1j_1}=(-1)^{i_1+j_1}\ket{i_1j_1}
\end{equation}
Substituting this action into Eq.~(\ref{y1relatmea}), we find
\begin{equation} \label{y1relatmea1}
    \sum_{i_1,j_1} (-1)^{i_1+j_1} \ket{i_1,j_1}\ket{\psi^{(k)}_{2,n}(i_{1},j_{1})}=\sum_{i_1,j_1}\ket{i_1,j_1}\ket{\psi^{(k)}_{2,n}(i_{1},j_{1})}
\end{equation}
Since the computational basis vectors $\ket{i_1j_1}$ are linearly independent, the coefficients of each basis element must match on both sides of the equality, i.e., for every pair $(i_1,j_1)$
\begin{equation}\label{y1relatmea1}
    (-1)^{i_1+j_1} \ket{i_1,j_1}\ket{\psi^{(k)}_{2,n}(i_{1},j_{1})}=\ket{i_1,j_1}\ket{\psi^{(k)}_{2,n}(i_{1},j_{1})}
\end{equation}
If $i_1+j_1$ is even (i.e., for even parity when both bits are equal, $i_1,j_1=00$ or $11$), $\ket{\psi^{(k)}_{2,n}(i_{1},j_{1})}$ remains unrestricted. If $i_1+j_1$ is odd (i.e., for odd parity when both bits differs, $i_1,j_1=01$ or $10$), $\ket{\psi^{(k)}_{2,n}(i_{1},j_{1})}=0$. Thus, the state has no support on the basis vectors corresponding to opposite computational outcomes between Alice and Bob. Only the even-parity components ($\ket{00}$ and $\ket{11}$) survive,
\begin{equation}
    \mathcal{U}^{(k)}\ket{\psi^{(k)}} = \ket{00}\otimes\ket{\psi^{(k)}_{2,n}(0,0)} + \ket{11}\otimes\ket{\psi^{(k)}_{2,n}(1,1)} 
\end{equation}
Applying the stabiliser for $y=2$, for the second Clifford generator,
\begin{equation}
    \Gamma_2^T\otimes\Gamma_2=\qty(\sigma_y \otimes - \sigma_y)\otimes \openone_{m_*/2} \otimes \openone_{m_*/2}
\end{equation}
and the relation
\begin{equation}
(\Gamma_2^T\otimes\Gamma_2) \sum_{i_1,j_1}\ket{i_1j_1}\ket{\psi^{(k)}_{2,n}(i_{1},j_{1})}=\sum_{i_1,j_1}\ket{i_1j_1}\ket{\psi^{(k)}_{2,n}(i_{1},j_{1})}
\end{equation}
yields the constraint $\ket{\psi^{(k)}_{2,n}(0,0)}=\ket{\psi^{(k)}_{2,n}(1,1)}$. Hence, the state in Eq.~(\ref{m1}) becomes
\begin{equation}
    \mathcal{U}\ket{\psi^{(k)}} = \qty(\ket{00}+\ket{11})\otimes\ket{\psi^{(k)}_{2,n}(0,0)}
\end{equation}
Up to normalisation, the first two qubits form a maximally entangled pair $\ket{\phi^+}=\qty(\ket{00}+\ket{11})/\sqrt{2}$. Note that for $n=2,3$, the state will reduce to a single copy of the maximally entangled two-qubit state as
\begin{equation}
    \mathcal{U}\ket{\psi^{(k)}} =  \ket{\phi^+}_{A'B'}\otimes\ket{junk}_{J_{A^{(k)}}J_{B^{(k)}}}
\end{equation}
where $\ket{junk}_{J_{A^{(k)}}J_{B^{(k)}}}=\sqrt{2}\ket{\psi^{(k)}_{2,n}(0,0)}$. Proceeding with the next anticommuting pair $\qty(\Gamma_4,\Gamma_5)$ acting on the subsequent tensor factors,
and repeating the same argument recursively, one finds that at each iteration $s$ (see Eq.~(\ref{simp}) and the discussion below it) the state factors an additional copy of $\ket{\phi^+}$,
\begin{equation}
    \mathcal{U}^{(k)}\ket{\psi^{(k)}} = \ket{\phi^+}^{\otimes \lfloor n/2 \rfloor} \otimes \ket{junk}_{J_A^{(k)}J_B^{(k)}}
\end{equation}
where $\ket{junk}_{J_A^{(k)}J_B^{(k)}} = 2^{\frac{1}{2}\lfloor \frac{n}{2} \rfloor}\ket{\psi(0,\dots,0)}_{J_A^{(k)}J_B^{(k)}}$ denotes a normalised ancillary state supported on the remaining degrees of freedom. 

Therefore, under the local unitary transformation $\mathcal{U}^{(k)}$, each Schmidt block $k$ of the optimal strategy is mapped onto $\mathcal{U}^{(k)}\ket{\psi^{(k)}} = \ket{\phi^+_{m_*}} \otimes \ket{junk}_{J_A^{(k)}J_B^{(k)}}$, where $m_*=2^{\lfloor n/2 \rfloor}$ is the minimal dimension supporting the irreducible Clifford representation.
Hence, each block contains a maximally entangled state on its irreducible support, tensored with uncorrelated ancillas that play no role in the Bell correlations.
\end{proof}

As an example, we provide a detailed explanation of self-testing of observables and state for the $n=4$ case in the Appx.~\ref{apxn4}.

The block-wise theorem guarantees the existence of local isometries acting within each Schmidt subspace that map the corresponding physical operators and states to their canonical Clifford-algebraic form. To extend this result to the entire Hilbert space, these block-wise isometries must be coherently combined into global maps acting on the full systems of Alice and Bob. The following lemma formalises this construction and shows that the direct sum of orthogonally defined isometries remains an isometry on the total space.

\begin{Lemma}[Direct-sum composition of local isometries] \label{dscli}
    Let $\qty{\Phi_A^{(k)}:\mathcal{H}_A^{(k)}\to \mathcal{H}_{A'}^{(k)} \otimes \mathcal{H}_{J_A^{(k)}}}_{k=1}^{\alpha}$ and $\qty{\Phi_B^{(k)}:\mathcal{H}_B^{(k)}\to \mathcal{H}_{B'}^{(k)} \otimes \mathcal{H}_{J_B^{(k)}}}_{k=1}^{\alpha}$ be families of local isometries defined on mutually orthogonal subspaces $\mathcal{H}_A=\bigoplus_{k=1}^{\alpha} \mathcal{H}_A^{(k)}$ and $\mathcal{H}_B=\bigoplus_{k=1}^{\alpha} \mathcal{H}_B^{(k)}$. Then the maps
    \begin{equation}\label{algiso}
     \Phi_A:=  \bigoplus_{k=1}^{\alpha} \Phi_A^{(k)} : \mathcal{H}_A \to \qty(\bigoplus_{k=1}^{\alpha} \mathcal{H}_{A'}^{(k)}) \otimes \qty(\bigoplus_{k=1}^{\alpha} \mathcal{H}_{J_A^{(k)}})
    \end{equation}
    and similarly
   \begin{equation}\label{bogiso}
     \Phi_B:=  \bigoplus_{k=1}^{\alpha} \Phi_B^{(k)} : \mathcal{H}_B \to \qty(\bigoplus_{k=1}^{\alpha} \mathcal{H}_{B'}^{(k)}) \otimes \qty(\bigoplus_{k=1}^{\alpha} \mathcal{H}_{J_B^{(k)}})
    \end{equation} 
    are themselves isometries satisfying $\Phi_A^{\dagger}\Phi_A=\openone_{\mathcal{H}_A}$ and $\Phi_B^{\dagger}\Phi_B=\openone_{\mathcal{H}_B}$.
\end{Lemma}

\begin{proof}
 We prove the claim for Alice’s side; the argument for Bob's side is identical. Since the full local Hilbert space $\mathcal{H}_A$ splits into orthogonal subspaces (the Schmidt blocks), $\mathcal{H}_A=\bigoplus_{k=1}^{\alpha}\mathcal{H}_A^{(k)}$, every vector $\ket{\psi_A}\in \mathcal{H}_A$ admits a unique orthogonal decomposition
 \begin{equation}
     \ket{\psi_A} = \sum_{k=1}^{\alpha} \ket{\psi_A^{(k)}}, \ \ \text{where } \ket{\psi_A^{(k)}} \in \mathcal{H}_A^{(k)}.
 \end{equation}
Because the decomposition is orthogonal, we have $\ip{\psi_A^{(k)}}{\psi_A^{(k')}}=\delta_{kk'}$. Now for two arbitrary vectors $\ket{\chi_A},\ket{\psi_A} \in \mathcal{H}_A$, their inner product in $\mathcal{H}_A$ is given by $\ip{\psi_A}{\chi_A}=\sum_{k,k'}\ip{\psi^{(k)}_A}{\chi^{(k')}_A}$. But because the subspaces $\mathcal{H}_A^{(k)}$ are orthogonal, $\ip{\psi_A}{\chi_A}$ simplifies to
\begin{equation}
    \ip{\psi_A}{\chi_A} = \sum_{k=1}^{\alpha} \ip{\psi^{(k)}_A}{\chi^{(k)}_A}.
\end{equation}
Now apply the composed map $\Phi_A=\bigoplus_k \Phi_A^{(k)}$. Each $\Phi_A^{(k)}$ acts independently on its own block:
\begin{equation}
    \Phi_A \ket{\psi_A}=\sum_{k=1}^{\alpha} \Phi_A^{(k)}\ket{\psi_A^{(k)}}.
\end{equation}
Computing the inner product of the images, we get
\begin{equation}
    \ip{\chi_A (\Phi_A)^{\dagger}}{\Phi_A \psi_A}=\sum_{k,k'}\ip{\chi_A^{(k)} (\Phi_A^{(k)})^{\dagger}}{\Phi_A^{(k')} \psi^{(k')}_A}
\end{equation}
But again, the images of distinct $k$’s lie in orthogonal subspaces of the target, so all cross-terms vanish, i.e., $\ip{\chi_A^{(k)} (\Phi_A^{(k)})^{\dagger}}{\Phi_A^{(k')} \psi^{(k')}_A}=\delta_{kk'}$. Therefore,
\begin{eqnarray}
    \ip{\chi_A (\Phi_A)^{\dagger}}{\Phi_A \psi_A}&=&\sum_{k}\ip{\chi_A^{(k)} (\Phi_A^{(k)})^{\dagger}}{\Phi_A^{(k)} \psi^{(k)}_A} \nonumber \\
    &=&\sum_k \ip{\chi_A^{(k)}}{\psi^{(k)}_A} \ \ \text{[Since $(\Phi_A^{(k)})^{\dagger}\Phi_A^{(k)}=\openone$]} \nonumber \\
    &=&\ip{\chi_A }{\psi_A}
\end{eqnarray}
Hence, it follows that $\qty(\Phi_A)^{\dagger}\Phi_A=\openone$. This shows that the global map $\Phi_A$ preserves the inner products on all of $\mathcal{H}_A$, i.e., $\Phi_A$ is an isometry. The same argument applies identically to $\Phi_B$.

\end{proof}

Having demonstrated that the local isometries associated with each Schmidt block can be coherently assembled into a single global construction, we now proceed to establish the full self-testing statement for the entire quantum realisation.

\begin{Corollary}[Global self-testing from the maximal violation of $\mathscr{G}_n$]
Let a bipartite quantum strategy $\qty{A_{x}, B_{y}, \ket{\psi}}$ achieve the maximal quantum value $\qty(G_n)^{opt}_{\mathcal{Q}}=2^{n-1}\sqrt{n}$ of the Bell functional $\mathscr{G}_n$, with the state is given by Eq.~(\ref{statescmidtgl}) and also the observables admit the block structure given by Eq.~(\ref{obsscmidtgl}) with properties given be Eq.~(\ref{gloobsprop}). Then there exist \textit{global} local isometries (up to complex conjugation) given by Eqs.~(\ref{algiso}) and (\ref{bogiso}), such that 
\begin{equation}\label{finalobsst}
\begin{aligned}
     \qty(\Phi_A\otimes\Phi_B)\ket{\psi} &= \ket{\phi^+_{m_*}}_{A'B'} \otimes \ket{\text{junk}}_{J_AJ_B}, \\
     \qty(\Phi_A \mathcal{A}_y \Phi_A^{\dagger}) \otimes \openone_{J_B}&=\Gamma_y^T \otimes \openone_{J_AJ_B} \ \ \forall y, \\
     \qty(\Phi_B B_y \Phi_B^{\dagger})\otimes \openone_{J_A}&=\Gamma_y \otimes \openone_{J_AJ_B} \ \ \forall y.
\end{aligned}
\end{equation}
where $\qty{\Gamma_y}$ are the canonical Clifford generators acting on the minimal irreducible subspace of dimension $m_*=\lfloor n/2\rfloor$. 
\end{Corollary}

\begin{proof}
  Expressing the physical state in Schmidt decomposition of Eq.~(\ref{statescmidtgl}), according to Lemma~\ref{dscli} the action of global isometry on the state will be action of each block isometry on each Schmidt block of the state, then substituting the block-self-testing relation from Theorem~\ref{bscr}, we obtain
  \begin{eqnarray}
     \qty(\Phi_A\otimes\Phi_B)\ket{\psi} &=&\bigoplus_{k=1}^{\alpha} \sqrt{\lambda^{(k)}m_k} \ket{\phi^+_{m_*}}_{A'B'} \otimes \ket{\text{junk}}_{J_{A^{(k)}}J_{B^{(k)}}} \nonumber \\
     &=&\ket{\phi^+_{m_*}}_{A'B'} \otimes \qty(\bigoplus_{k=1}^{\alpha} \sqrt{\lambda^{(k)}m_k} \ket{\text{junk}}_{J_{A^{(k)}}J_{B^{(k)}}}) \nonumber \\
     &=& \ket{\phi^+_{m_*}}_{A'B'} \otimes \ket{\text{junk}}_{J_AJ_B}.
  \end{eqnarray}
  Also, since the observables are block-diagonal, the global isometries act component wise
  \begin{equation}
  \begin{aligned}
      \qty(\Phi_A \mathcal{A}_y \Phi_A^{\dagger}) &= \bigoplus_{k=1}^{\alpha} \Gamma_y^T \otimes \openone_{J_A^{(k)}} \\
      &=\Gamma_y^T \otimes \openone_{J_A}
  \end{aligned}
  \end{equation}
  and similar relation holds for Bob. Therefore, after applying the global local isometries, the physical realisation is locally equivalent to a reference experiment consisting of a maximally entangled state $\ket{\phi^+_{m_{*}}}$ of minimal dimension and identical Clifford observables on both sides, tensored with arbitrary ancillary systems that play no role in the observed Bell correlations. 
\end{proof}


\section{Robust self-testing}\label{sec5}

In the preceding Sec.~\ref{selftesting}, we have established the existence of local isometries, $\Phi_A$ and $\Phi_B$, that extracts the target quantum strategy $\qty{A'_x,B'_y,\ket{\psi'}}$ from the ideal one $\qty{A_x,B_y,\ket{\psi}}$. In the absence of noise, the optimal quantum violation of the Bell functional, $\expval{\mathscr{G}_n}^{opt}=2^{n-1}\sqrt{n}$ is achieved for this ideal strategy.

In practice, however, the observed violation is typically slightly below the optimal value, $\expval{\mathscr{G}_n}^{obs}=2^{n-1}\sqrt{n}-\delta$ for some small $\delta>0$. This deviation reflects that the physical strategy, of unknown dimension, denoted by $\qty{\tilde{A}_x,\tilde{B_y},\ket{\tilde{\psi}}}$, does not exactly satisfy the ideal relations $M_x\ket{\psi}=0$. Instead, the corresponding relations are violated by a controlled amount, with  $\norm{\tilde{M}\ket{\tilde{\psi}}}$ bounded by a function of $\delta$. Without loss of generality, the physical state can be considered to be pure and the observables Hermitian and dichotomic (i.e., of eigenvalues $\pm 1$), since no assumption is made on the Hilbert space dimension, which can always be extended to purify the state and make the measurements projective.

The local isometry then maps these physical objects $\qty{\tilde{A}_x,\tilde{B_y},\ket{\tilde{\psi}}}$ to a target reference strategy of known dimension, $\qty{\tilde{A}'_x,\tilde{B}'_y,\ket{\tilde{\psi'}}}$. In the ideal (noise-free) case, one would expect $\qty{A'_x=\tilde{A}'_x,B'_y=\tilde{B}'_y,\ket{\psi'}=\ket{\tilde{\psi'}}}$, corresponding to $\delta=0$ (up to complex conjugation).

However, since $\delta>0$ in reality, there exists a finite deviation between the mapped physical observables and states and their ideal target counterparts. In the following, we quantify these deviations by deriving the operator norm bounds, which characterise how close the extracted strategy $\qty{\tilde{A}'_x,\tilde{B}'_y,\ket{\tilde{\psi'}}}$ is to the corresponding target reference quantities $\qty{A'_x,B'_y,\ket{\psi'}}$ as a function of the observed discrepancy $\delta$. 

Specifically, we evaluate the following norm bounds
\begin{equation}
    \norm{\openone\otimes \qty(B'_y-\tilde{B}'_y)\ket{\tilde{\psi}'}}, \ \ \norm{\qty(A'_x-\tilde{A}'_x)\otimes \openone \ket{\tilde{\psi}'}}, \ \ \norm{\ket{\tilde{\psi}'}-\ket{\tilde{\psi}}}
\end{equation}
Since every dichotomic POVM can be dilated to a projective observable on a larger space, without loss of generality we assume that all physical observables are Hermitian with spectrum contained in $\qty{\pm1}$, i.e. projective dichotomic observables, i.e., $\norm{\tilde{B}_y\ket{\tilde{\psi}}}=\norm{\tilde{A}_x\ket{\tilde{\psi}}}=1$.

{\thm \label{thm5} Let $\qty{\tilde{A}_x, \tilde{B}_y,\ket{\tilde{\psi}}}$ be a physical strategy achieving
\begin{equation} \label{suboptrob}
    \expval{\mathscr{G}_n}_{\tilde{\psi}}= 2^{n-1}\sqrt{n}-\delta, \ \ \delta\geq 0,
\end{equation}
then there exist constants $C_n,D_n,E_n\geq 0$, such that the following approximate relations hold on the support of $\ket{\tilde{\psi}}$
\begin{equation}\label{robth}
\begin{aligned}
(i) & \ \ \norm{\qty(\Phi_A\otimes \Phi_B) \ket{\tilde{\psi}}- \ket{\psi'}}  \leq C_n \sqrt{\delta} \\
(ii) & \ \ \norm{\qty(\Phi_A\tilde{A}_x\Phi_A^{\dagger}-A'_x)\otimes \openone \ket{\tilde{\psi}'}}  \leq D_n \sqrt{\delta} \\
 (iii) & \ \ \norm{\openone\otimes \qty(\Phi_B\tilde{B}_y\Phi_B^{\dagger}-B'_y) \ket{\tilde{\psi}'}}  \leq E_n \sqrt{\delta}
\end{aligned}
\end{equation}
and the pairwise anticommutators of Bob’s observables satisfy
\begin{equation}
\abs{\expval{\qty{\tilde{B}'_y,\tilde{B}'_{y'}}}_{\tilde{\psi}}}\leq L_n \sqrt{\delta}
\end{equation}
where $\ket{\psi'}=\ket{\phi^{+}_{m_*}}\otimes \ket{junk}$. $A'_x$, and $B'_y$ are the target observables given by Eq.~(\ref{finalobsst}). Moreover, the constants scale as $C_n\sim \order{n^{\frac{-1}{4}}}$, and $D_n,E_n\sim \order{n^{\frac{1}{4}}}$.}

\begin{proof}
For the ideal scenario, the SOS decomposition guarantees
\begin{equation}\label{idsos}
  \forall x \ \ \   M_x\ket{\psi}=0 \implies A_x \otimes\openone \ket{\psi}=\frac{1}{\sqrt{n}}(\openone\otimes\mathcal{B}_x)\ket{\psi},
\end{equation}
establishing a perfect correspondence between Alice’s and Bob’s ideal observables on the support of the ideal state $\ket{\psi}$. Under the isometry $\Phi:=\Phi_A\otimes\Phi_B$, this relation is preserved on the target reference state $\ket{\psi'}$. Importantly, the condition $M_x\ket{\psi}=0 \ \forall x$ implies that $\ket{\psi}$ lies in the subspace $\mathcal{S}:=\bigcap_x \ker(M_x)$, i.e., the intersection of the kernels of all SOS operators $M_x$. Any state achieving the maximal violation must belong to this subspace, i.e., $\ket{\psi}\in \mathcal{S}$. 

For a physical realisation $\ket{\tilde{\psi}}$ achieving a slightly suboptimal value $\expval{\mathscr{G}_n}_{\tilde{\psi}}=\expval{\mathscr{G}_n}_{\psi}-\delta$, a component necessarily lies outside $\mathcal{S}$, and the distance from the kernel quantifies the deviation of the physical state from the ideal one. Taking the expectation value of the SOS decomposition in Eq.~(\ref{starti}) with respect to $\ket{\tilde{\psi}}$ gives
\begin{equation} \label{sosrobmain}
\expval{G_n \openone - \mathscr{G}_n}_{\tilde{\psi}} = \sum_{x = 1}^{2^{n-1}} \frac{\omega_x}{2} \expval{ M_x^{\dagger} M_x }_{\tilde{\psi}}
\end{equation}
which can be equivalently expressed as, by taking $\expval{G_n \openone - \mathscr{G}_n}_{\tilde{\psi}}=2^{n-1}\sqrt{n}-\expval{\mathscr{G}_n}_{\tilde{\psi}}=\delta$ and using Eq.~(\ref{suboptrob})
\begin{equation}
    \delta=\sum_{x = 1}^{2^{n-1}} \frac{\omega_x}{2} \norm{M_x\ket{\tilde{\psi}}}^2
\end{equation}
Since each term $\frac{\omega_x}{2} \norm{M_x\ket{\tilde{\psi}}}^2$ in the sum is non-negative, each contribution must be individually bounded by the total deviation $\delta$. Using the fact that $\omega_x=\sqrt{n}$ in the ideal case and remains bounded in the near-optimal regime, one obtains
\begin{equation} \label{mbound}
\norm{M_x\ket{\tilde{\psi}}} \leq F_n \sqrt{\delta}  \  \forall x, \ \ F_n=\sqrt{\frac{2^{2-n}}{\sqrt{n}}}
\end{equation}
This immediately provides a bound on the component of the physical state lying outside the kernel of the SOS operators. Physically, this allows $\ket{\tilde{\psi}}$ to be decomposed into a component within the ideal subspace $\mathcal{S}$, denoted by $\ket{\psi_{\parallel}}\in \mathcal{S}$, and an orthogonal component $\ket{\psi_{\perp}}\in \mathcal{S}^{\perp}$, given by
\begin{equation}\label{orperc}
    \ket{\psi_{\parallel}}=\frac{\Pi\ket{\tilde{\psi}}}{\norm{\Pi\ket{\tilde{\psi}}}},  \ \ \ \ket{\psi_{\perp}}=\frac{\qty(\openone-\Pi)\ket{\tilde{\psi}}}{\norm{\qty(\openone-\Pi)\ket{\tilde{\psi}}}}
\end{equation}
where $\Pi$ is the orthogonal projector on to $\mathcal{S}$. Decomposing $\ket{\tilde{\psi}}$ as
\begin{equation}
    \ket{\tilde{\psi}} = \Pi \ket{\tilde{\psi}} + \qty(\openone-\Pi)\ket{\tilde{\psi}}
\end{equation}
we find (Appx.~\ref{statebound} for detailed derivation)
\begin{equation}\label{staterobfinal}
        \norm{ \ket{\tilde{\psi}}-\ket{\psi}} =\norm{\qty(\openone-\Pi)\ket{\tilde{\psi}}}
\leq C_n \sqrt{\delta}, \ C_n\sim\order{n^{-1/4}}
\end{equation}
Since the isometry $\Phi$ acts locally and preserves vector norms, the above bounds remain valid after mapping the physical state to the reference space.
\begin{equation}
         \norm{ \ket{\tilde{\psi}'}-\ket{\psi}'} \leq C_n \sqrt{\delta}, \ C_n\sim\order{n^{-1/4}}
\end{equation}
For the physical implementation, the norm of Bob’s effective observable operator deviates slightly from its ideal value. In the ideal case, the quantity $\omega_x=\norm{\mathcal{B}_x}_{\psi}=\sqrt{n}$. For the physical realisation, we denote the corresponding (deviated) quantity as
$\tilde{\omega}_x:=\norm{\tilde{\mathcal{B}}_x}_{\tilde{\psi}}$. Hence, we obtain
\begin{equation}\label{omro}
\tilde{\omega}_x^2=\expval{\tilde{\mathcal{B}}_x^2}_{\tilde{\psi}}=n+\sum_{y<y'} (-1)^{z^x_y+z^x_{y'}} \expval{\qty{\tilde{B}_y,\tilde{B}_{y'}}}_{\tilde{\psi}} = n+\Delta_x \ \ \ \forall x
\end{equation}
We define
\begin{equation}\label{dx}
    \Delta_x:=\sum_{y<y'} (-1)^{z^x_y+z^x_{y'}} \eta_{yy'}, \ \ \ \eta_{yy'}:= \expval{\qty{\tilde{B}_y,\tilde{B}_{y'}}}_{\tilde{\psi}} \ \ \forall y\neq y'
\end{equation}
The quantity $\Delta_x=\tilde{\omega}_x^2-n$ thus quantifies the deviation of $\tilde{\omega}_x$ from its ideal value. In the ideal case, where all observables $B_y$ perfectly anticommute, then $\eta_{yy'}=0$ and consequently $\Delta_x=0$. In Appx.~\ref{roancobb}, we derive a bound on the degree of non-anticommutativity for Bob’s observables, given by
\begin{equation}
    \abs{\expval{\qty{\tilde{B}'_y,\tilde{B}'_{y'}}}_{\tilde{\psi}}} \leq L_n \sqrt{\delta}, \ \ \ L_n \sim \order{\frac{1}{2^{n+1}n}}
\end{equation}

To quantify the robustness of the observables, we begin by decomposing the deviation of Alice’s physical observable from its ideal counterpart into two parts: one correlated with Bob’s observables and a residual term capturing the remaining difference. The key idea is that the SOS relations allow Bob’s scaled operators to be written as linear combinations of Alice’s observables on the support of the state. Near-optimal violation therefore implies that deviations in the SOS relations translate directly into bounded deviations of the individual measurement operators. Using standard norm inequalities, this yields a bound on Alice’s observable in terms of the residual error operator and the deviation of Bob’s scaled observables from their ideal forms as (see Appx.~\ref{roobs} for detailed evaluation)
\begin{equation}\label{alicerobim1}
  \norm{\qty(\tilde{A}_x-A_x)\otimes \openone}_{\ket{\tilde{\psi}}} \leq \qty(F_n+Q_n) \sqrt{\delta} +\norm{\openone\otimes \frac{\tilde{\mathcal{B}}_x-\mathcal{B}_x}{\tilde{\omega}_x}}_{\ket{\tilde{\psi}}}
\end{equation}

Explicitly, by expressing Bob’s scaled physical observable $\tilde{\mathcal{B}}_x$ in terms of Alice’s physical observable $\tilde{A}_x$ and a residual error term $\tilde{M}$, one finds that, in the ideal noiseless case, Bob’s ideal scaled observable $\mathcal{B}_x$ can be represented as a linear combination of Alice’s physical operators on the subspace defined by the ideal state $\ket{\psi}$. Then, Applying the triangle inequality and exploiting the SOS relations, we obtain (see Appx.~\ref{roobs} for detailed evaluation)
\begin{equation}\label{alicerobim2}
     \frac{1}{\tilde{\omega}_x}\norm{\openone\otimes\qty(\tilde{\mathcal{B}}_x-\mathcal{B}_x)\ket{\tilde{\psi}}} \leq  H_n \sqrt{\delta},
\end{equation}
where $H_n\sim \order{n^{-\frac{1}{4}}}$. Then, combining Eqs.~(\ref{alicerobim1}) and (\ref{alicerobim2}), we obtain
\begin{equation}\label{alicerobim3}
  \norm{\qty(\tilde{A}_x-A_x)\otimes \openone \ket{\tilde{\psi}}} \leq D_n \sqrt{\delta} \ \ \ \text{with } D_n \sim \order{n^{\frac{1}{4}}}
\end{equation}
By inverting the linear relation defining the scaled observables given by Eq.~(\ref{curlyBx}), each of Bob’s physical operators can be expressed as a linear combination of the deviations of the scaled observables. Norm inequalities then yield explicit robustness bounds for Bob’s individual measurements from Eq.~(\ref{alicerobim1}) as
\begin{equation}\label{bobrobim3}
\norm{\openone\otimes\qty(\tilde{B}_y-B_y)\ket{\tilde{\psi}}} \leq E_n \sqrt{\delta}  \ \ \ \text{with } E_n \sim \order{n^{\frac{1}{4}}}
\end{equation}

\end{proof}

The above bounds establish that any strategy achieving a near-maximal violation of the Bell functional must be close, up to local isometries, to the ideal quantum realisation identified in Sec.~\ref{selftesting}. In particular, the extracted state approaches the maximally entangled target state while the implemented measurements approximate the corresponding Clifford observables, with deviations scaling at most as $\order{\sqrt{\delta}}$. These results demonstrate that the self-testing statement derived for the ideal case remains effective under small experimental imperfections.


\section{Summary and Discussion}\label{secSD}

In this work, we propose a scalable self-testing framework for the maximally entangled two-qudit state of local dimension $m_*=2^{\lceil \frac{1}{2}(n-1) \rceil}$ (equivalently $\lceil \frac{1}{2}(n-1) \rceil$ copies of maximally entangled two-qubit states \cite{Kraft2018}), together with $n$ mutually anticommuting observables on one side. The framework is based on an $n$-settings, two-outcome Bell inequality. We first derived analytically the local ontic bound of this inequality for arbitrary $n$. We then determine the optimal quantum bound of the Bell functional $\mathscr{G}_n$ by employing the SOS decomposition technique \cite{Bamps2015}, without assuming any restriction on the underlying Hilbert space dimension. The resulting quantum bound exceeds the classical one for all $n$, and its analytical form enables a complete characterisation of the quantum strategies saturating it. 

The maximal quantum value is achieved by a state that decomposes into a direct sum of maximally entangled components, each supported on local subspaces of dimension at least $m_*$. On every such Schmidt block, both parties implement identical families of dichotomic, pairwise anticommuting Hermitian observables forming a representation of the complex Clifford algebra $\mathrm{Cl}_n\qty(\mathbb{C})$. Consequently, the full measurement operators are block-diagonal, with each block acting irreducibly on the corresponding Schmidt subspace and carrying an independent Clifford representation. Any nontrivial degeneracy of the Schmidt spectrum therefore corresponds to classical mixing between identical Clifford blocks.

Within each block, Alice's $2^{n-1}$ observables correspond geometrically to the vertices of an $n$-dimensional hypercube inscribed in the operator Bloch sphere defined by the Clifford generators, $\mathcal{A}_y^{(k)}$, cf. by Eq.~(\ref{aliceobsy}). Each observable direction can be represented by the normalised vector $\vec{v}_x=\frac{1}{\sqrt{n}}\qty((-1)^{z^x_1},\ldots,(-1)^{z^x_n})$, encoding the pattern of signs defining its relation to the $n$ anticommuting axes.  Bob’s $n$ measurements correspond to the orthogonal Clifford directions themselves, the coordinate axes of this hypercube, while Alice’s observables form its vertices. The hypercube symmetry ensures that the averaged correlations are invariant under any parity inversion of the coefficients, thereby realising the parity-oblivious constraints relations known from the $n\to 1$ parity-oblivious multiplexing task in a prepare-measure scenario \cite{Spekkens2009, Ghorai2018, Singh2025njp}. 

A further subtle feature of $\mathscr{G}_n$ arises from the fact that the optimal quantum strategy is not unique under local unitaries. In particular, if the strategy $\qty{\mathcal{A}_y=\Gamma_y^T,B_y=\Gamma_y,\ket{\phi^+}}$ attains the maximal quantum value, then its complex-conjugated counterpart $\qty{\mathcal{A}^T_y=\Gamma_y,B_y^T=\Gamma_y^T,\ket{\phi^+}}$ does as well. This equivalence occurs because the transposition (or complex conjugation) operation maps a given irreducible representation of $\mathrm{Cl}_n\qty(\mathbb{C})$ to another inequivalent but correlation-preserving one, for all Hermitian combinations $\Gamma_y^T\otimes\Gamma_y$, the corresponding joint expectation values remain identical. Consequently, any two-outcome Bell functional composed purely of joint correlators with real coefficients, such as $\mathscr{G}_n$, cannot distinguish between a strategy and its complex-conjugated counterpart. Physically, this means that the Bell violation is invariant under local complex conjugation, and that any self-testing statement based on the optimal violation of such two-outcome Bell inequalities must therefore be understood up to both local isometries and local complex conjugation. This generalises the ambiguity previously identified for the Elegant Bell inequality \cite{Gisin2007, Andersson2017} (recovered as the special case $n=3$), showing that it is not an isolated feature but a general algebraic consequence of Bell functionals whose optimal observables realise complex Clifford-algebra representations.

Within this extended equivalence framework, we have proved the existence of local isometries that extract, from any optimal physical realisation, the reference configuration composed of maximally entangled states of local dimension $m_*$ and $n$ mutually anticommuting measurement operators, up to local unitary and complex-conjugation freedoms. This formally establishes that the maximal quantum value of $\mathscr{G}_n$ self-tests a tensor product of maximally entangled qubit pairs (of total local dimension $m_*=2^{\lfloor n/2 \rfloor}$) together with the associated Clifford measurement structure.

We have also analysed the robustness of this self-testing scheme by examining deviations from the SOS kernel. In the ideal case, the optimal state lies entirely within the kernel subspace of the SOS operators. When the observed value deviates from the quantum maximum by a small parameter $\delta$, the physical state necessarily acquires a small component orthogonal to this kernel. Operationally, this means that the experimentally realised correlations include a residual contribution orthogonal to the ideal algebraic constraints defining the optimal quantum strategy. The distance of the physical state from this kernel subspace thus provides a direct quantitative measure of how far the observed statistics are from the ideal self-tested configuration. By bounding this orthogonal component, we obtained explicit norm inequalities quantifying how close the extracted physical strategy is to the ideal reference strategy. In particular, we show that the physical state and observables remain close (in operator norm) to their ideal counterparts within an error of order $\order{n^{1/4}}\sqrt{\delta}$. Thus, even as the number of measurement settings increases, the degradation in fidelity grows only polynomially with $n$, demonstrating that the protocol remains stable in the large-setting regime. In addition, we established that the anticommutation relations of Bob’s observables remain approximately preserved as $\order{a/2^n n}\sqrt{\delta}$, confirming that the Clifford structure is robust under small perturbations. Although these robustness bounds are derived analytically under certain approximations, they can be systematically tightened using the techniques of \cite{Kaniewski2019}, enabling quantitative benchmarking of different scalable self-testing schemes.

Beyond their theoretical relevance, our results have direct implications for device-independent quantum information processing. Because the inequality $\expval{\mathscr{G}_n} \leq (G_n)^{opt}_{\mathcal{L}}$ self-tests multiple copies of maximally entangled pairs using only two-outcome measurements, it naturally lends itself to applications in DI randomness expansion and quantum key distribution. The possibility of certifying an unbounded amount of randomness through a scalable number of measurement settings paves the way towards asymptotically high performance DI quantum protocols.\\

\section{Acknowledgements}
SS acknowledges the support from the National Natural Science Fund of China (Grant No. W2533013). RKS acknowledges the financial support from the Council of Scientific and Industrial Research (CSIR, 09/1001(17051)/2023-EMR-I), Government of India. PR acknowledges the support by KIAS individual Grant No. QP100601 at the Korea Institute for Advanced Study and by the local hospitality from the grant IITH/SG160 of IIT Hyderabad, India.  AKP acknowledges the support from the Research Grant SERB/CRG/2021/004258, Government of India.


\appendix
\onecolumngrid

\section{Proof of theorem 1}\label{thm1}

 The local realist bound of the $n$-settings Bell functional (Eq.~(10) of the main text) is given by
\begin{equation} \label{bnlb1}
	\qty|\expval{\mathscr{G}_{n}}_{L}|\leq  \sum\limits_{x=1}^{2^{n-1}} \qty| \  \sum\limits_{y=1}^{n} (-1)^{z_y^{x}} \expval{B_y}_{\lambda} \ | \leq \sum\limits_{x=1}^{2^{n-1}} \qty| \  \sum\limits_{y=1}^{n} (-1)^{z_y^{x}} \ |
\end{equation}
Note that $-1\leq \expval{B_y}\leq 1$ implies the optimal local realist bound of $ \mathscr{G}_{n}$ occurs when $\expval{B_y} \in \{+1,-1\}$. Interestingly, regardless of whether $\expval{B_y}=+1$ or $\expval{B_y}=-1$, the value of $|\expval{\mathscr{G}}_L|$ remains unchanged. 

Since $z_y^x \in\{0,1\}$, the quantity $(-1)^{z_y^{x}}$ is either $0$ or $1$. The value of $z_y^x \in\{0,1\}$ is fixed by the term appearing at $x^{th}$ row of $y^{th}$ column of the following matrix $\mathcal{S}_{(2^{n-1}\times n)}$. Subsequently, the value of $(-1)^{z_y^{x}}$ is also fixed by the matrix $\mathcal{M}_{(2^{n-1}\times n)}$.

\begin{eqnarray}
\mathcal{S}_{(2^{n-1}\times n)}=	
	\begin{bmatrix} 
		0 & 0 & \cdots& 0 & 0\\
		0 & 0 & \cdots& 0 & 1\\
			0 & 0 & \cdots& 1 & 0\\
		\vdots& \vdots&\ddots&\vdots&\vdots\\
	0 & 1 & \cdots& 1 & 1
	\end{bmatrix}  \ ; \ \ 
\mathcal{M}_{(2^{n-1}\times n)}=	
	\begin{bmatrix} 
		1 & 1 & \cdots& 1 & 1\\
		1 & 1 & \cdots& 1 & -1\\
		1 & 1 & \cdots& -1 & 1\\
		\vdots& \vdots&\ddots&\vdots&\vdots\\
	1 & -1 & \cdots& -1 & -1
	\end{bmatrix} 
\end{eqnarray}

Modulus of sum of each row of $\mathcal{M}_{(2^{n-1}\times n)}$ provides $ \qty| \  \sum\limits_{y=1}^{n} (-1)^{z_y^{x}}  \ | \ \forall x$. For instance, the sum of the elements in the first row is $ \qty| \  \sum\limits_{y=1}^{n} (-1)^{z_y^{x=1}}  \ | =n$, on the other hand, the sum is $ \qty| \  \sum\limits_{y=1}^{n} (-1)^{z_y^{x=2^{n-1}}}  \ | = |(2-n)| = (n-2)$. 

Now, in order to evaluate the sum for any given row, we must determine the count of elements in $\mathcal{M}_{(2^{n-1}\times n)}$ that are either $+1$ or $-1$ in that specific row. Note that the first element in each row (corresponding to the first column) is set to $+1$. Now, lets consider that in the $r^{th}$ row, out of remaining $(n-1)$ elements, there are $k$ number of elements having value $-1$. Then sum of the elements in this particular row becomes $(n-2k)$. Furthermore, as there are $\binom{n-1}{k}$ rows that contain $k$ occurrences of $-1$ entries, the total sum of the elements for such rows is $(n-2k)\binom{n-1}{k}$. It's worth mentioning that $(n-2k)>0$ when $0\leq k \leq \lfloor n/2 \rfloor$. Consequently, the overall sum of rows containing more elements with $+1$ value than those with $-1$ value is given by $\sum\limits_{k=0}^{\lfloor n/2 \rfloor} (n-2k)\binom{n-1}{k}$.

For rows containing fewer elements with $+1$ than with $-1$ values, the sum is assessed by counting the elements with $+1$ values. Let there be $k'$ occurrences of $+1$ elements in the $(r')^{th}$ row. The sum for this specific row is then $(2k'-n)$. Additionally, there are $\binom{n-1}{k'-1}$ rows that contain $k'$ instances of $+1$ entries. Consequently, the sum for such rows is $(2k'-n)\binom{n-1}{k-1}<0 \ \forall 0 \leq k' \leq \lfloor n/2 \rfloor$.   

Therefore, following the preceding arguments, we obtain
\begin{eqnarray} \label{bnlb2}
(\mathscr{G}_n)_{\mathcal{L}} = \sum\limits_{x=1}^{2^{n-1}} \qty| \ \sum\limits_{y=1}^{n} (-1)^{z_{y}^{x}} \ | &=& \sum\limits_{k=0}^{\lfloor \frac{n}{2} \rfloor}  \binom{n-1}{k} \ (n-2k) + \qty|\sum\limits_{k'=1}^{\lfloor \frac{n}{2} \rfloor}  \binom{n-1}{k'-1} \ (2k'-n)| \nonumber \\
&=& \sum\limits_{k=0}^{\lfloor \frac{n}{2} \rfloor}  \binom{n-1}{k} \ (n-2k) + \sum\limits_{k'=1}^{\lfloor \frac{n}{2} \rfloor}  \binom{n-1}{k'-1} \ (n-2k') \\
&=& n + \sum\limits_{k=1}^{\lfloor \frac{n}{2} \rfloor}  \qty[\binom{n-1}{k}+\binom{n-1}{k-1}] \ (n-2k)
\end{eqnarray} 
Now, changing the index $k'$ to $k$ and using the binomial identity $\qty[\binom{n-1}{k}+\binom{n-1}{k-1}]=\binom{n}{k}$, we get the following:
\begin{eqnarray}
(\mathscr{G}_n)_{\mathcal{L}}&=& \binom{n}{0} \ n + \sum\limits_{k=1}^{\lfloor \frac{n}{2} \rfloor} \binom{n}{k} \ (n-2k) =\sum\limits_{k=0}^{\lfloor \frac{n}{2} \rfloor} \binom{n}{k} \ \qty{(n-k)-k} \\
&=& \sum\limits_{k=0}^{\lfloor \frac{n}{2} \rfloor} \frac{n!}{(n-k-1)! \ k!} - \sum\limits_{k=1}^{\lfloor \frac{n}{2} \rfloor} \frac{n!}{(n-k)! \ (k-1)!} \\
&=& \sum\limits_{k=0}^{\lfloor \frac{n}{2} \rfloor} \frac{n!}{(n-k-1)! \ k!} - \sum\limits_{l=0}^{\lfloor \frac{n}{2} \rfloor-1} \frac{n!}{(n-l-1)! \ l!} \\
&=& \frac{n!}{(n-\lfloor n/2 \rfloor-1)! \ \lfloor n/2 \rfloor!} =\binom{n}{\lfloor n/2 \rfloor+1} \ \qty(\lfloor n/2 \rfloor+1)
\end{eqnarray}      
which is given in Eq. (11) of the main text.


\section{Evaluation of the optimal quantum bound of the Bell functional $\mathscr{G}_n$ via the SOS Method}\label{detailedSOS}

We evaluate the optimal quantum value of the Bell functional $\mathscr{G}_n$ employing the sum-of-squares (SOS) relaxation method. This approach reformulates the non-convex optimisation over quantum states and measurement operators into a convex problem, enabling analytic characterisation of both the optimal bound and the structure of the corresponding quantum realisations.

Applying SOS relaxation, the shifted Bell operator can be written as
\begin{equation} \label{sstart}
    {G}_n \openone_n - \mathscr{G}_n = \sum_{x = 1}^{2^{n-1}} \frac{\omega_x}{2} \ M^{\dagger}_x M_x,
\end{equation}
where $M_x$ and $\omega_x$ are defined by
\begin{equation}\label{somegaa}  
M_{x} := \frac{\openone\otimes \mathcal{B}_x}{\omega_{x}} - A_x \otimes \openone \ \ ; \ \ \omega_{x} := \sqrt{\Tr[\mathcal{B}_x^{\dagger}\mathcal{B}_x \ \rho]} \geq 0, \ \ \text{and }  \mathcal{B}_x=\sum\limits_{y=1}^{n} (-1)^{z^x_y} B_y.
\end{equation}
Using Eq.~(\ref{somegaa}) and assuming dichotomic observables $A_x^2=B_y^2=\openone$ (since their eigenvalues are $\pm 1$), we find
\begin{equation} \label{sstart}
    M_x^\dagger M_x=\sum\limits_{x=1}^{2^{n-1}}\qty[ \frac{\qty(\mathcal{B}_x)^2}{\omega_x^2} +  \openone] -2\sum\limits_{x=1}^{2^{n-1}}\frac{A_x \otimes \mathcal{B}_x}{\omega_x}.
\end{equation}
Substituting $M_x$ and $\omega_{x}$ from Eq.~(\ref{somegaa}) into the right-hand side of Eq.~(\ref{sstart}) yields
\begin{equation} \label{sbello}
\sum_{x = 1}^{2^{n-1}} \frac{\omega_x}{2} \ M^{\dagger}_x M_x= \frac{1}{2}\sum\limits_{x=1}^{2^{n-1}}\qty[ \frac{\qty(\mathcal{B}_x)^2}{\omega_x} +  \omega_x \ \openone ] -\sum\limits_{x=1}^{2^{n-1}} A_x \otimes \mathcal{B}_x = \frac{1}{2}\sum\limits_{x=1}^{2^{n-1}}\qty[ \frac{\qty(\mathcal{B}_x)^2}{\omega_x} +  \omega_x \ \openone] - \mathscr{G}_{n}
\end{equation}
Comparing this with the SOS definition and taking the expectation value with respect to $\rho$ gives
\begin{eqnarray}\label{sbetasos1a}
	G_{n}&=&\frac{1}{2}\sum\limits_{x=1}^{2^{n-1}} \frac{1}{\omega_x} \Tr[\qty(\mathcal{B}_x)^2  \rho ] + \frac{1}{2} \sum\limits_{x=1}^{2^{n-1}} \omega_{x} ,  \nonumber \\
	&=& \sum\limits_{x=1}^{2^{n-1}} \omega_x,  \ \ \ \text{since } \omega_{x}^2 =\Tr[\mathcal{B}_x^{\dagger}\mathcal{B}_x \rho]  
\end{eqnarray}
Thus, the optimal quantum bound is
\begin{equation}\label{ssosopt}
    (G_n)_{\mathcal{Q}}^{opt}= \max \sum_{x=1}^{2^{n-1}} \omega_x = \max\sum_{x=1}^{2^{n-1}} \sqrt{\Tr[\mathcal{B}_x^{\dagger}\mathcal{B}_x \rho]}
\end{equation}
Expanding $\mathcal{B}_x^{\dagger}\mathcal{B}_x$ produces contributions both from both the diagonal terms $B_y^2$ and the cross terms $B_yB_{y'}$. Taking the expectation with respect to $\rho$ and using $B_y^2=\openone$, we obtain
\begin{equation}\label{omega1}
     \omega_x^2=n+\sum_{y< y'}   (-1)^{z^x_y+z^x_{y'}} \Tr[\qty{B_y, B_{y'}}\rho]
\end{equation}


\subsection{Evaluation of $\Tr[\qty{B_y,B_{y'}}\rho] \ \ \forall y \neq y'$} \label{appacp}

For all $\forall 1\leq y < y' \leq n$, define the real vectors
\begin{equation}\label{landc}
     \mathbf{\Lambda}:=\qty[\Lambda_{12},\Lambda_{13},\ldots,\Lambda_{(n-1)n}]^{T}\in \mathbb{R}^{\frac{n(n-1)}{2}}, \ \ \text{with } \Lambda_{yy'}:=\Tr[\qty{B_y, B_{y'}}\rho].
\end{equation}
and, for each $x \in \qty{1,2,\ldots,2^{n-1}}$,
\begin{equation}\label{landc1}
     \mathbf{C}_x:=\qty[c^x_{(1,2)},c^x_{(1,3)},\ldots, c^x_{(n-1,n)}]\in \mathbb{R}^{\frac{n(n-1)}{2}}, \ \ \text{with } c^x_{(y,y')}:= (-1)^{z^x_y+z^x_{y'}} \in \qty{+1,-1}.
\end{equation}
With these definitions, the optimisation in Eq.~(\ref{ssosopt}) becomes
\begin{equation}\label{ssosopt1}
    (G_n)_{\mathcal{Q}}^{opt} = \max \sqrt{n+\mathbf{C}_x \cdot \mathbf{\Lambda}}.
\end{equation}

\begin{Lemma} \label{idenlemsup}
 Let $z^x=(z_1^x,z_2^x,\ldots,z_n^x)$ denote binary strings of length $n$, where each component $z^x_y \in \qty{0,1}$, $z^x_1=0$ fixed and the remaining $n-1$ bits varying freely, yielding $2^{n-1}$ distinct strings. Then, for all $x\in \qty{1,\cdots,2^{n-1}}$ and $y,y'\in \qty{1,\cdots,n}$, the following identities hold
    \begin{eqnarray}
          &\text{(i)}& \sum_{x=1}^{2^{n-1}} (-1)^{z^x_y+z^x_{y'}} = \begin{cases}
              0 & \text{if } y \neq y' \\
              2^{n-1} & \text{if } y = y'
          \end{cases}\label{iden1} \\
               &\text{(ii)}& \sum_{x=1}^{2^{n-1}} (-1)^{z^x_{y_1}+z^x_{y'_1}} \cdot (-1)^{z^x_{y_2}+z^x_{y'_2}} =  \begin{cases} 
               0 & \text{if } (y_1,y'_1)\neq (y_2,y'_2) \\
               2^{n-1} & \text{if } (y_1,y'_1) = (y_2,y'_2)
          \end{cases} \label{iden2}
    \end{eqnarray}
\end{Lemma}

\begin{proof}
    
Each index $x$ corresponds to a binary string $z^x=\qty(z_1^x,z_2^x,\ldots,z_n^x)\in \{0,1\}^n$ with $z^x_1=0$. The complete set of such binary strings is
\begin{equation}
V=\qty{z=(0,z_2,z_3,\ldots,z_n)\ | \ z_i \in \qty{0,1} }, \ \  \abs{V}=2^{n-1}
\end{equation}
Any sum over the index $x$ can equivalently be expressed as a sum over the set $V$, i.e.
\begin{equation}
 \sum_{x=1}^{2^{n-1}}f(z^x)=\sum_{z\in V} f(z)
\end{equation}
for any function $f:V\to \mathbb{R}$. The order of summation is immaterial, as these are finite sums independent of indexing.

To prove Eqs.~(\ref{iden1}) and (\ref{iden2}), we employ the Walsh-Hadamard orthogonality. A Walsh function (or Walsh character) is defined for any subset $S\subseteq \qty{2,3,\ldots,n}$ as
\begin{equation}
    \chi_S : V \to \qty{+1,-1}, \ \ \chi_S=(-1)^{\sum_{i\in S}z_i}
\end{equation}
Each $\chi_S$ maps a binary string $z\in V$ to $\pm 1$, depending on the parity of bits indexed by $S$. When $S=\emptyset$, $\chi_{\emptyset}(z)=1$ is the trivial Walsh function. The collection $\qty{\chi_S:S\subseteq \{2,3,\ldots,n\}}$ forms a complete orthogonal basis for real-valued functions on $V$. Specifically,
\begin{equation}
    \sum_{z\in V} \chi_{S}(z) \chi_{S'}(z)=
    \begin{cases}
        2^{n-1} & \text{if } S=S' \\
        0 & \text{if } S\neq S'
    \end{cases}
\end{equation}
Now, for any pair $(y,y')$, define the subset $S_1=\qty{y,y} \cap S$. Then, 
\begin{equation}
    (-1)^{z_y+z_{y'}}=(-1)^{\sum_{i\in S_1}z_i}=\chi_{S_1}(z),
\end{equation}
since $z_1=0$ contributes nothing whenever $y=1$ or $y'=1$. Summing over all $z\in V$,
\begin{equation}
    \sum_{z\in V} (-1)^{z_y+z_{y'}}= \sum_{z\in V} \chi_{S_1}(z)
\end{equation}
Applying the orthogonality of Walsh characters gives
\begin{equation}
    \sum_{z\in V} \chi_{S_1}(z) =
    \begin{cases}
        2^{n-1}, & \text{if } S_1=\emptyset \\
        0, & \text{if } S_1\neq \emptyset
    \end{cases}
\end{equation}
The subset $S_1=\emptyset$ precisely when $y=y'$, thereby proving Eq.~(\ref{iden1}).

Next, for any two index pairs $\qty(y_1,y'_1)$ and $\qty(y_2,y'_2)$,
\begin{equation}
    (-1)^{z_{y_1}+z_{y'_1}} \cdot (-1)^{z_{y_2}+z_{y'_2}}=(-1)^{z_{y_1}+z_{y'_1}+z_{y_2}+z_{y'_2}}.
\end{equation}
We define another subset $S_2 = \qty{y_1,y'_1,y_2,y'_2} \cap S$ so that so that the product again corresponds to a Walsh character $\chi_{S_2}(z)$. Then,
\begin{equation}
   \sum_{z\in V}(-1)^{z_{y_1}+z_{y'_1}} \cdot (-1)^{z_{y_2}+z_{y'_2}}=\sum_{z \in V} \chi_{S_2}(z)=
   \begin{cases}
       2^{n-1}, & \text{if } S_2=\emptyset \\
       0, & \text{if } S_2\neq \emptyset
   \end{cases}
\end{equation}
The subset $S_2=\emptyset$ only when the two index pairs coincide, i.e. $\qty(y_1,y'_1)=\qty(y_2,y'_2)$, which establishes Eq.~(\ref{iden2}).
\end{proof}

Expanding Eq.~(\ref{ssosopt1}) in a Taylor series gives
\begin{equation}\label{taylor1}
   \begin{aligned}
      (G_n)^{opt}_{\mathcal{Q}} &= \max \sum_{x=1}^{2^{n-1}} \sqrt{n+\qty(\mathbf{C}_x \cdot \mathbf{\Lambda})} = \max \sum_{x=1}^{2^{n-1}} \qty{\sqrt{n}\sum_{k=0}^{\infty}\binom{1/2}{k}\frac{\qty(\mathbf{C}_x \cdot \mathbf{\Lambda})^k}{n^k}} \\
      &=\max \sum_{x=1}^{2^{n-1}} \qty{\sqrt{n}+\frac{\mathbf{C}_x \cdot \mathbf{\Lambda}}{2\sqrt{n}}-\frac{\qty(\mathbf{C}_x \cdot\mathbf{\Lambda})^2}{8 n^{\frac{3}{2}}} + \cdots}
   \end{aligned} 
\end{equation}
Considering the linear term $\sum_x \mathbf{C}_x\cdot\mathbf{\Lambda}$ from Eq.~(\ref{taylor1}), we get
\begin{equation} \label{expansup1}
\begin{aligned}
 \sum_{x=1}^{2^{n-1}} \mathbf{C}_x\cdot\mathbf{\Lambda} &= \sum_{x=1}^{2^{n-1}} \sum_{y<y'} c^x_{yy'} \ \Lambda_{yy'} \ \ \ [\text{From Eq.~(\ref{landc})}] \\
 &=\sum_{y<y'} \qty(\sum_{x=1}^{2^{n-1}} c^x_{yy'}) \ \Lambda_{yy'} =0 \ \ \ [\text{From Eq.~(\ref{iden1})}].
\end{aligned}
\end{equation}
Similarly, considering the quadratic term $\sum_x \qty(\mathbf{C}_x \cdot \mathbf{\Lambda})^2$ from Eq.~(\ref{taylor1}), we find
\begin{equation}\label{expansup2}
\begin{aligned}
 \sum_{x=1}^{2^{n-1}} \qty(\mathbf{C}_x \cdot \mathbf{\Lambda})^2 &= \sum_{x=1}^{2^{n-1}} \qty(\sum_{y<y'} c^x_{yy'} \ \Lambda_{yy'})^2 \ \ \ [\text{From Eq.~(\ref{landc})}] \\
 &= \sum_{x=1}^{2^{n-1}} \qty{\sum_{y_1<y'_1} \sum_{y_2<y'_2} c^x_{y_1y'_1}c^x_{y_2y'_2} \ \Lambda_{y_1y'_1} \ \Lambda_{y_2y'_2}} \\
 &=\sum_{y_1<y'_1} \sum_{y_2<y'_2} \qty(\sum_{x=1}^{2^{n-1}} c^x_{y_1y'_1}c^x_{y_2y'_2})  \ \Lambda_{y_1y'_1} \ \Lambda_{y_2y'_2} \\
 &=\sum_{y_1<y'_1} \sum_{y_2<y'_2} \Lambda_{y_1y'_1} \ \Lambda_{y_2y'_2} \qty(2^{n-1}\delta_{(y_1,y'_1),(y_2,y'_2)}) \ \ \ [\text{From Eq.~(\ref{iden2})}] \\ 
 &= 2^{n-1} \sum_{y\leq y'} \Lambda_{yy'}^2
\end{aligned}
\end{equation}
More generally, for an arbitrary power $k$,
\begin{equation}\label{taylor2}
\begin{aligned}
   \sum_{x=1}^{2^{n-1}} \qty(\mathbf{C}_x\cdot\mathbf{\Lambda})^k &=  \sum_{x=1}^{2^{n-1}} \qty(\sum_{y_1<y'_1} \cdots \sum_{y_k<y'_k} \Lambda_{y_1y'_1} \cdots \Lambda_{y_ky'_k} \ c^x_{(y_1<y'_1)} \cdots c^x_{(y_k<y'_k)}) \\
   &=\sum_{y_1<y'_1} \cdots \sum_{y_k<y'_k} \Lambda_{y_1y'_1} \cdots \Lambda_{y_ky'_k} \qty(\sum_{x=1}^{2^{n-1}} c^x_{(y_1<y'_1)} \cdots c^x_{(y_k<y'_k)})
\end{aligned}
\end{equation}
By the same reasoning as in Lemma~\ref{idenlemsup}, and using Walsh–Hadamard orthogonality, the inner sum over $x$ is non-zero only when every index appears with even multiplicity (i.e. all pairs are repeated). Therefore,
\begin{equation}\label{iden3}
    \sum_{x=1}^{2^{n-1}} \qty(\mathbf{C}_x\cdot\mathbf{\Lambda})^k =
    \begin{cases}
        0, & \text{if } k=(2m+1) \\
        2^{n-1} \sum\limits_{y<y'} \Lambda^{2m}_{yy'}, & \text{if } k=2m
    \end{cases}
\end{equation}
The result in Eq.~(\ref{iden3}) can be understood as a direct consequence of parity symmetry within the Walsh–Hadamard basis. Each coefficient $c^x_{(y,y')}$ alternates sign according to the parity of the corresponding bits in the binary string $z^x$. When taking products of such terms, every index $y$ contributes a factor of $(-1)^{z^x_y}$ each time it appears. If an index appears an even number of times, these factors cancel pairwise, resulting in a net contribution of $+1$; if it appears an odd number of times, a residual $(-1)^{z^x_y}$ remains. Summing over all binary strings $z^x$ uniformly, any residual sign depending on $z^x_y$ averages to zero due to the symmetry between configurations with $z^x_y \in \{0,1\}$. Consequently, only those terms where all indices appear in even multiplicities, i.e., all pairs are repeated, survive the summation. This explains why all odd-order terms vanish and why, for even orders, only diagonal contributions proportional to $\Lambda^{2m}_{yy'}$ remain.

Substituting Eq.~(\ref{iden3}) into Eq.~(\ref{taylor1}), we obtain
\begin{equation}\label{taylorf}
   (G_n)^{opt}_{\mathcal{Q}} = F(\mathbf{\Lambda}):= \max \qty[ 2^{n-1}\sqrt{n} - 2^{n-1}\sqrt{n} \sum_{m=1}^{\infty}\frac{1}{n^m} \ \binom{1/2}{m} \ \qty(\sum\limits_{y<y'} \Lambda^{2m}_{yy'})]
\end{equation}
Note that the real valued function $F(\mathbf{\Lambda})$ is strictly concave, as its Hessian satisfies $H(\mathbf{\Lambda})=\laplacian{\mathbf{\Lambda}}<0$ throughout the interior of the compact convex domain $\mathcal{D}\subset \mathbb{R}^{\frac{n(n-1)}{2}}$. Strict concavity ensures a unique global maximum, attained where $\grad{\mathbf{\Lambda}}=0$, i.e., at $\Lambda_{yy'}=0 \ \forall y\neq y'$. 

Moreover, the point $\mathbf{\Lambda}=0$ is itself physically attainable: If Bob’s observables form a set of $n$ pairwise anti-commuting dichotomic operators acting irreducibly on a Hilbert space of dimension $2^{\lfloor n/2\rfloor}$ (for instance, a representation of the complex Clifford algebra), and $\rho=\dyad{\phi^{+}}$ is a maximally entangled state on this space, then indeed $\Tr[\qty{B_y,B_{y'}}\rho]=0$.

Evaluating $F$ at the optimal point $\mathbf{\Lambda}=0$ yields
\begin{equation}\label{sacoppr}
 (G_n)^{opt}_{\mathcal{Q}} = 2^{n-1}\sqrt{n}, \ \ \text{with } \Lambda_{yy'}=\Tr[\qty{B_y,B_{y'}}\rho]=0 \ \ \forall y<y' \in \qty{1,2,\ldots,n} 
\end{equation}
Thus, the maximal achievable value of the Bell functional $\mathscr{G}_n$ is attained if and only if there exists an optimal state $\rho$ for which all pairwise anticommutators vanish, i.e.,  $\Tr[\qty{B_y,B_{y'}}\rho]=0 \ \forall y\neq y'$, consequently $\omega_x=\sqrt{n} \ \forall x \in [2^{n-1}]$, and hence $\omega_x=\sqrt{n}, \ \forall x \in \qty[2^{n-1}]$.


\subsection{Observable relations} 

Taking the expectation value of the SOS identity given in Eq.~(\ref{sstart}) with the optimal state $\rho$, we obtain
\begin{equation}\label{orcos1}
     \Tr[ \qty({G}_n \openone_n - \mathscr{G}_n)\rho] = \sum_{x = 1}^{2^{n-1}} \frac{\omega_x}{2} \Tr[ M^{\dagger}_x M_x \rho]
\end{equation}
If the upper bound is saturated, i.e. $\Tr[\mathscr{G}_n \rho]=(G_n)^{opt}_Q=2^{n-1}\sqrt{n}$, then the left-hand side of Eq.~(\ref{orcos1}) vanishes, implying
\begin{equation}\label{orcos2}
\frac{\sqrt{n}}{2} \sum_{x = 1}^{2^{n-1}} \Tr[ M^{\dagger}_x M_x \rho]=0
\end{equation}
Since both $(M_x)^{\dagger}M_x$ and $\rho$ are positive semidefinite, each term $\Tr[ M^{\dagger}_x M_x \rho]$ is non-negative. Consequently, the sum can vanish only if every term is individually zero
\begin{equation}\label{orcos3}
    \Tr[ M^{\dagger}_x M_x \rho]=0 \ \ \ \forall x.
\end{equation}
Let $\ket{\psi}$ be any purification of $\rho$ on $\qty(\mathcal{H}_A\otimes \mathcal{H}_{J_A})\otimes\qty(\mathcal{H}_B\otimes \mathcal{H}_{J_B}) $, where $\mathcal{H}_{J_A},\mathcal{H}_{J_B}$ are suitable auxiliary Hilbert spaces, and $M_x$ acts on the system Hilbert space only. Then,
\begin{equation}\label{mdm1}
     \Tr[ M^{\dagger}_x M_x \rho]=\expval{M^{\dagger}_x M_x \otimes \openone}{\psi}=\norm{M_x \otimes \openone \ket{\psi}}^2
\end{equation}
By Eq.~(\ref{orcos3}), the norm on the right is zero, and omitting tensor factors for clarity we obtain
\begin{equation}\label{msizero}
M_x\ket{\psi}=0 \ \ \ \forall x.
\end{equation}
Substituting the expression for $M_x$ from Eq.~(\ref{somegaa}), we find
\begin{equation}\label{aborf}
   \qty(\openone \otimes \mathcal{B}_x) \ket{\psi}=\sqrt{n}  \qty(A_x\otimes\openone) \ket{\psi} \ \ \ \forall x.
\end{equation}
Acting with $\qty(A_x\otimes\openone)$ on left of this Eq.~(\ref{aborf}) gives
\begin{equation}\label{aborf1}
   \qty(A_x \otimes \mathcal{B}_x) \ket{\psi}=\sqrt{n} \ket{\psi} \ \ \ \forall x.
\end{equation}
Similarly, acting with $\qty(\openone \otimes \mathcal{B}_x)$ on the left of Eq.~(\ref{aborf}) yields
\begin{equation}\label{aborf2}
 \forall x \ \   \qty(\openone \otimes \mathcal{B}_x^2) \ket{\psi}=n \ket{\psi} \implies \sum_{y<y'}c^x_{yy'} \qty(\openone \otimes \qty{B_y,B_{y'}}) \ket{\psi}=0.
\end{equation}
where $c^x_{yy'}$ is defined in Eq.~(\ref{landc1}). The columns of 
the matrix $\mathbf{C}=\qty(c^x_{yy'})$ are mutually orthogonal and nonzero, as follows from Eq.~(\ref{iden2}) of the lemma~\ref{idenlemsup}, namely $\sum_x c^x_{yy'}c^x_{y_1y'_1}=2^{n-1}\delta_{(y,y'),(y_1,y'_1)}$. Because the linear map $\qty(\openone \otimes \qty{B_y,B_{y'}}) \ket{\psi} \mapsto \sum_{y<y'}c^x_{yy'} \qty(\openone \otimes \qty{B_y,B_{y'}}) \ket{\psi}$ is injective, one may multiply Eq.~(\ref{aborf2}) by $c^x_{y_0y'_0}$ and sum over all $x$ without affecting equality,
\begin{equation}\label{aborf3}
    \sum_x c^x_{y_0y'_0} \sum_{y<y'}c^x_{yy'} \qty(\openone \otimes \qty{B_y,B_{y'}}) \ket{\psi}=0 \implies \sum_{y<y'}c^x_{yy'} \qty(\sum_x c^x_{y_0y'_0} c^x_{yy'})  \qty(\openone \otimes \qty{B_y,B_{y'}}) \ket{\psi}=0. 
\end{equation}
Using, $\sum_x c^x_{yy'}c^x_{y_1y'_1}=2^{n-1}\delta_{(y,y'),(y_1,y'_1)}$, we obtain
\begin{equation}\label{finacb}
   2^{n-1} \qty(\openone \otimes \qty{B_{y_0},B_{y'_0}}) \ket{\psi}=0 \implies  \qty(\openone \otimes \qty{B_{y},B_{y'}}) \ket{\psi}=0 \ \ \forall y<y'.
\end{equation}
Thus, for the optimal violation, Bob's observables must be mutually anticommuting on the support of the optimal state $\ket{\psi}$.

To derive the corresponding relations for Alice’s observables, consider two distinct indices $x\neq x'$. Expanding $\mathcal{B}_x$ and $\mathcal{B}_{x'}$ in terms of $B_y$ and using Eq.~(\ref{aborf}), we have
\begin{equation}
    \qty(A_x\otimes\openone) \ket{\psi}=\frac{1}{\sqrt{n}} \qty(\openone \otimes \sum_{y=1}^n (-1)^{z_y^x} B_y) \ket{\psi}, \ \ \ \qty(A_{x'}\otimes\openone) \ket{\psi}=\frac{1}{\sqrt{n}} \qty(\openone \otimes \sum_{y'=1}^n (-1)^{z_{y'}^x} B_{y'}) \ket{\psi}
\end{equation}
Then,
\begin{equation}\label{alicac1}
\begin{aligned}
 \qty(A_x A_{x'} \otimes\openone) \ket{\psi} &= \frac{1}{n} \qty(\openone \otimes\sum_{y,y'} (-1)^{z_y^x+z_{y'}^{x'}}  B_y B_{y'}) \ket{\psi} \\
 &=\frac{2}{n} \qty(\openone \otimes \sum_{y=y'}(-1)^{z_y^x+z_{y'}^{x'}} B_yB_{y'}) \ket{\psi} + \frac{2}{n} \qty(\openone \otimes \sum_{y<y'} (-1)^{z_y^x+z_{y'}^{x'}}  B_y B_{y'}) \ket{\psi} \\
 &= \frac{2}{n} \qty(\openone \otimes \sum_{y}(-1)^{z_y^x+z_{y}^{x'}} \openone) \ket{\psi} + \frac{2}{n} \qty(\openone \otimes \sum_{y<y'} (-1)^{z_y^x+z_{y'}^{x'}}  B_y B_{y'}) \ket{\psi}
 \end{aligned}
\end{equation}
Similarly,
\begin{equation}\label{alicac2}
 \qty(A_{x'} A_x \otimes\openone) \ket{\psi} = \frac{2}{n} \qty(\openone \otimes \sum_{y}(-1)^{z_y^x+z_{y}^{x'}} \openone) \ket{\psi} + \frac{2}{n} \qty(\openone \otimes \sum_{y<y'} (-1)^{z_y^x+z_{y'}^{x'}} B_{y'} B_y ) \ket{\psi}
\end{equation}
Using the fact that $\qty(\openone \otimes \qty{B_{y},B_{y'}}) \ket{\psi}_{y\neq y'}=0$ for all $y \neq y'$, and adding Eqs.~(\ref{alicac1}) and (\ref{alicac2}), we obtain
\begin{equation}
   \qty(\qty{A_x,A_{x'}}\otimes\openone)\ket{\psi}=\frac{2}{n} \sum_{y} (-1)^{z_y^x+z_{y}^{x'}} \ket{\psi} \ \ \ \forall x\neq x' \in [2^{n-1}].
\end{equation}
Hence, the anticommutator of Alice’s observables is proportional to the Hamming dot product of the corresponding binary strings $z_y^x$ and $z_{y}^{x'}$.

Now, consider the quantity $\sum\limits_{x=1}^{2^{n-1}} (-1)^{z_y^x} A_x$. From Eq.~(\ref{aborf}), we find
\begin{equation}
\begin{aligned}
     \sum\limits_{x=1}^{2^{n-1}} (-1)^{z_y^x} A_x \otimes \openone \ket{\psi} &= \frac{1}{\sqrt{n}}  \sum\limits_{x=1}^{2^{n-1}}\sum\limits_{y'=1}^n (-1)^{z_y^x+z_{y'}^x} \qty(\openone\otimes B_{y'}) \ket{\psi} \\
    &= \frac{2^{n-1}}{\sqrt{n}} \sum\limits_{y'=1}^{n} \delta_{y,y'} \qty(\openone\otimes B_{y'}) \ket{\psi}=\frac{2^{n-1}}{\sqrt{n}} \qty(\openone\otimes B_y)\ket{\psi}
\end{aligned}
\end{equation}
Define the normalised, scaled observables for Alice as
\begin{equation} \label{aliceobsy}
    \mathcal{A}_y:=\frac{\sqrt{n}}{2^{n-1}} \sum\limits_{x=1}^{2^{n-1}} (-1)^{z_y^x} \ A_x \ \ \ \forall y \in [n].
\end{equation}
Then the relationship between Alice’s and Bob’s observables on the support of $\ket{\psi}$ can be re-expressed, using Eq.~(\ref{aborf}), as
\begin{equation} \label{baorf}
   \qty(\openone\otimes B_y)\ket{\psi}=  \qty(\mathcal{A}_y \otimes \openone) \ket{\psi} \implies \ev**{\qty(\mathcal{A}_y \otimes B_y)}{ \psi}=1 \ \ \ \forall y\in[n].
\end{equation}
Thus, on the support of the optimal state, Bob’s operators act as the transpose (or conjugate counterparts) of Alice’s.


\subsection{Structure of the Optimal State}\label{optimalState}

Let $\qty{\ket{i}_A},\qty{\ket{j}_B}$ be fixed orthonormal bases of the local Hilbert spaces $\mathcal{H}_A$ and $\mathcal{H}_B$, respectively. Define $\mu_{ij}\in \mathbb{C}$, satisfying $\abs{\mu_{ij}}\leq 1$ and $\sum_{ij}\abs{\mu_{ij}}^2 =1$, as the complex coefficients obtained by projecting $\ket{\psi}$ onto each product basis element, i.e., $\mu_{ij}=\mel{i}{\psi}{j}$. Then, any pure bipartite state $\ket{\psi} \in \mathcal{H}_A \otimes \mathcal{H}_B$ admits the decomposition
\begin{equation} \label{stader1}
    \ket{\psi}=\sum_{i,j=1}^d \mu_{ij} \ket{i}_A \otimes \ket{j}_B=\mathrm{vec}\qty(\mathscr{M})
\end{equation}
where the vectorisation map $\mathrm{vec}: \mathbb{C}^{d \times d} \to \mathbb{C}^{d^2}$ is defined as
\begin{equation}
 \mathrm{vec}\qty(\mathscr{M}):=\sum_{i,j=1}^d \mu_{ij} \ket{i}_A \otimes \ket{j}_B,
\end{equation}
thereby associating the bipartite vector $\ket{\psi}$ with the $d\times d$ coefficient matrix $\mathscr{M}:=[\mu_{ij}]\in \mathbb{C}^{d \times d}$, i.e. $\mathscr{M}_{ij}=\mu_{ij}$.

There always exists local unitary $U_A$ and $V_B$ acting on $\mathcal{H}_A$ and $\mathcal{H}_B$, respectively, that transform the local bases to the Schmidt basis, in which $\mathscr{M}$ becomes diagonal and positive. This follows from the singular value decomposition, 
\begin{equation}
 \mathscr{M}=U_A \mathcal{D} V_B^{\dagger}, \ \   \mathcal{D}=\begin{bmatrix}
        D_{r \times r} & \emptyset_{r \times d-r} \\ 
        \emptyset_{d-r \times r} & \emptyset_{d-r \times d-r}
    \end{bmatrix},
\end{equation}
where $D_{r \times r}=\text{diag}\qty(\lambda_1,\ldots,\lambda_r)$ contains the Schmidt coefficients $\sqrt{\lambda_i} > 0$ with $\sum_i \lambda_i =1$, and $r=\mathrm{rank}\qty(\mathscr{M})$. The state in Eq.~(\ref{stader1}) then can be expressed as
\begin{equation}\label{statevec}
    \ket{\psi}=\mathrm{vec}\qty(\mathscr{M})= \mathrm{vec}\qty(U_A \mathcal{D} V_B^{\dagger}) =\qty(U_A \otimes V^{\dagger}) \ \mathrm{vec}\qty(\mathcal{D}^{\frac{1}{2}}), \ \ \ \text{with } \mathrm{vec}\qty(\mathcal{D}^{\frac{1}{2}}) = \sum_{i=1}^r \sqrt{\lambda_i} \ket{u_i} \otimes \ket{v_i}.
\end{equation}
Since,
\begin{equation}\label{phivec}
\mathrm{vec}\qty(\mathcal{D}^{\frac{1}{2}}) = \mathrm{vec}\qty(D^{\frac{1}{2}})=\sqrt{r} \qty(\openone \otimes D^{\frac{1}{2}}) \ket{\phi^{+}}, \ \ \text{with } \ket{\phi^{+}}=\frac{1}{\sqrt{r}} \sum_i \ket{i}_A \ket{i}_B,
\end{equation}
we arrive at the following form
\begin{equation}\label{statevec2}
    \ket{\psi}=\sqrt{r}\qty(U_A \otimes V^{\dagger}_B) \qty(\openone \otimes \mathcal{D}^{\frac{1}{2}}) \ket{\phi^{+}}.
\end{equation}
Now, applying the optimal condition from Eq.~(\ref{baorf}) and using the standard identity $\qty(A \otimes B) \mathrm{vec}\qty(Y)=\mathrm{vec}\qty(A Y B^T)$, we find
\begin{equation}\label{lhsschb}
    (\mathcal{A}_y \otimes \openone)\ket{\psi}=\mathrm{vec}\qty(\mathcal{A}_yU_A\mathcal{D}^{\frac{1}{2}}V_B^{\dagger})=(\openone \otimes B_y)\ket{\psi} = \mathrm{vec}\qty(U_A \mathcal{D}^{\frac{1}{2}} V_B^{\dagger} B_y^{T})
\end{equation}
Injectivity of the vectorisation map implies
\begin{equation}\label{soprel}
  \mathcal{A}_yU_A\mathcal{D}^{\frac{1}{2}}V_B^{\dagger} = U_A \mathcal{D}^{\frac{1}{2}} V_B^{\dagger} B_y^{T} \implies \qty(U_A^{\dagger}\mathcal{A}_y U_A) \mathcal{D}^{\frac{1}{2}} = \mathcal{D}^{\frac{1}{2}} \qty(V_B^{\dagger} B_y^T V_B)
\end{equation}
We now redefine the observables in the Schmidt basis as
\begin{equation}
    \mathcal{A}_y \equiv U_A \mathcal{A}_y U_A^{\dagger}, \ \ B_y\equiv V_B^{\dagger}B^T_y V_B, 
\end{equation}
and define the Schmidt basis vectors by 
\begin{equation}
    \ket{u_i}:=U_A \ket{i}_A, \ \ \ \ket{v_i}:=V_B \ket{i}_B, \ \ \ \text{with } i \in \qty{1,2,\ldots,r}.
\end{equation}
Since $\mathcal{A}_y,B_y$ act on $\mathcal{H}^d$ locally, we extend the Schmidt basis by adding $d-r$ orthogonal vectors 
\begin{equation}\label{dschmidtb}
    \qty{\ket{u_p}}_{p=1}^d := \qty{\qty{\ket{u_i}}_{i=1}^r,\qty{\ket{u_i}}_{i=r+1}^d}; \ \ \ \qty{\ket{v_q}}_{q=1}^d := \qty{\qty{\ket{v_j}}_{j=1}^r,\qty{\ket{v_j}}_{j=r+1}^d}.
\end{equation}
In this basis, any Hermitian involutive operator acting on $\mathcal{H}^d$ can be expressed as a block matrix respecting the decomposition $\mathbb{C}^d=\mathbb{C}^r \oplus \mathbb{C}^{d-r}$
\begin{equation} \label{obsdchb}
    \mathcal{A}_y = \begin{bmatrix}
        \mathcal{A}_y^r & X_A \\
        X^{\dagger}_A & Z_A
    \end{bmatrix}, \ \ \ B_y = \begin{bmatrix}
        B_y^r & X_B \\
        X^{\dagger}_B & Z_B
    \end{bmatrix},
\end{equation}
where $\mathcal{A}_y^r,B_y^r \in \mathbb{C}^{r \times r}$ act on the respective Schmidt subspaces $(r \times r)$, and $X\in \mathbb{C}^{r \times(d-r)}$ and $Z\in \mathbb{C}^{(d-r) \times (d-r)}$ are the off-diagonal and complementary blocks. Moreover, since $\mathcal{A}_y$ and $B_y$ are Hermitian with $\mathcal{A}_y^2=B_y^2=\openone$, the following block relations hold
\begin{equation}\label{propblo}
    \qty(\mathcal{A}_y^r)^2+X_A X^{\dagger} = \openone_r, \ \ X^{\dagger}X_A + Z_A^2 = \openone_{d-r}, \ \ X^{\dagger} \mathcal{A}_y^r + Z_A X^{\dagger} = \mathcal{A}_y^r X_A + X_A Z_A=0.
\end{equation}
Now, in the extended Schmidt basis, the optimal condition of Eq.~(\ref{soprel}) becomes
\begin{equation}\label{abdaism}
    \mathcal{A}_y \mathcal{D}^{\frac{1}{2}} = \mathcal{D}^{\frac{1}{2}}B_y \implies \begin{bmatrix}
        \mathcal{A}_y^r D^{\frac{1}{2}} && \emptyset \\
       X_A^{\dagger} D^{\frac{1}{2}} && \emptyset
    \end{bmatrix} =
    \begin{bmatrix}
        D^{\frac{1}{2}} B_y^r  && D^{\frac{1}{2}} X_B \\
        \emptyset && \emptyset
    \end{bmatrix} \implies
    \begin{cases}
        (i) & \mathcal{A}_y^r = D^{\frac{1}{2}} B_y^r D^{-\frac{1}{2}}  \\
        (ii) & X_b=X_A^{\dagger} =0
    \end{cases}
\end{equation}
Thus, both $\mathcal{A}_y$ and $b_y$ are block diagonal in the extended Schmidt basis. From Eq.~(\ref{propblo}), we further have $\qty(\mathcal{A}_y^r)^2=\qty(B_y^r)^2=\openone_r$. Using $\qty(\mathcal{A}_y^r)^2=\openone_r$, we find
\begin{equation} \label{bdcom}
  \openone=\qty(\mathcal{A}_y^r){\dagger} \mathcal{A}_y^r  =  \qty(D^{\frac{1}{2}}B_y^r D^{-\frac{1}{2}})^{\dagger}\qty(D^{\frac{1}{2}}B_y^r D^{-\frac{1}{2}}) = D^{-\frac{1}{2}}B_y^r D B_y^r  D^{-\frac{1}{2}} \implies B_y^r D B_y^r =D \implies \qty[B_y^r,D]=0 \ \ \forall y \in [n].
\end{equation}
Using $\qty(B_y^r)^2=\openone_r$, A symmetric argument yields $\qty[\mathcal{A}_y^r,D]=0 \ \ \forall y \in [n]$. 

Let the distinct eigenvalues of $D$ be denoted by $\lambda^{(k)}$ with multiplicities $m_k$ for $k=1,\ldots,\alpha$, satisfying $\sum_{i=1}^{\alpha} m_{k}=r$. Then, $D=\oplus_{k=1}^{\alpha} \lambda^{(k)} \openone_{m_k}$, where $\openone_{m_k}$ is the identity on the $m_k$-dimensional degeneracy subspace, and the state can be expressed as
\begin{equation}
    \ket{\psi}=\sum_{k=1}^{\alpha} \sqrt{\lambda^{(k)}} \sum_{i \in I_k} \ket{u_i} \otimes \ket{v_i}=\bigoplus\limits_{k=1}^{\alpha} \sqrt{\lambda^{(k)} m_k} \ket{\phi^{+}_{m_k}}, \ \ \ \text{with } \ \ket{\phi^{+}_{m_k}}:=\frac{1}{\sqrt{m_k}} \sum_{i \in I_k} \ket{u_i} \otimes \ket{v_i}
\end{equation}
where each $\ket{\phi^{+}_{m_k}} \in \mathcal{H}^{m_k}_A \otimes \mathcal{H}^{m_k}_B$ is a maximally entangled on the $k$-th degeneracy block. The observables can be expressed in the Schmidt basis as
\begin{equation}
 \forall y, \  \mathcal{A}_y=\sum_{i,j}a_{ij}\dyad{u_i}{u_j}, \ \ a_{ij}= \matrixel{u_i}{\mathcal{A}_y}{u_j}, \ \ \ \text{and }  \  B_y=\sum_{i,j}b_{ij}\dyad{v_i}{v_j}, \ \ b_{ij}= \matrixel{v_i}{B_y}{v_j}.
\end{equation}
Since all algebraic relations are valid only on the support of $\ket{\psi}$, whose rank is $r$, we may, without loss of generality, restrict our attention to the subspace $\mathcal{H}^r_A \otimes \mathcal{H}^r_A$. For notational simplicity, we omit the superscript $r$ henceforth when referring to states, observables, or Hilbert spaces.

Using the relation obtained in Eq.~(\ref{bdcom}), with straightforward algebra yields
\begin{equation}
    [B_y,D]=0 \implies \matrixel{v_i}{[B_y,D]}{v_j}=\qty(\lambda_j-\lambda_i)b_{ij} =0 \ \ \forall i,j
\end{equation}
It follows that whenever $\lambda_i \neq \lambda_j$, the corresponding coefficient $b_{ij}=0$ for for every pair $i,j$. Hence, each $B_y$ has non-zero entries only within subspaces associated with equal Schmidt coefficients. The same argument applies to $\mathcal{A}_y$ in the $\ket{u_i}$ basis. Therefore, both sets of observables $\mathcal{A}_y$ and $B_y$ are block-diagonal with respect to the degeneracy structure of $D$
\begin{equation}
 \forall y, \  \mathcal{A}_y = \bigoplus\limits_{k=1}^{\alpha} \mathcal{A}_y^{(k)}, \ \ \ \text{and }  B_y = \bigoplus\limits_{k=1}^{\alpha} B_y^{(k)}, \ \ \ \mathcal{A}_y^{(k)}, B_y^{(k)} \in \mathscr{L}\qty(\mathcal{H}^{(k)}),
\end{equation}
where each block corresponds to a degenerate Schmidt value $\lambda_{i_1}=\lambda_{i_2}=\cdots=\lambda_{i_{m_k}}$, acting on the subspace $\mathcal{H}^{(k)}=\mathrm{span}\qty{\ket{u_{i_1}}:i \in I_k}$. In matrix form, all entries between different eigenspaces vanish, rendering each observable block-diagonal.

Using the optimality condition $\mathcal{A}_y D^{\frac{1}{2}}=D^{\frac{1}{2}} B_y$ from Eq.~(\ref{abdaism}), and noting that $D^{\frac{1}{2}}$ acts as a scalar multiple of the identity on each degenerate subspace, i.e. $\eval{D^{\frac{1}{2}}}_{\mathcal{H}^{(k)}}=\sqrt{\lambda^{(k)}} \openone_{m_{k}}$, we get
\begin{equation}
 \forall y, \  \mathcal{A}_y^{(k)}\sqrt{\lambda^{(k)}} \openone_{m_{k}}=\sqrt{\lambda^{(k)}} \openone_{m_{k}}B_y^{(k)} \implies \mathcal{A}_y^{(k)} = B_y^{(k)} \ \ \ \text{for every block $k$.}
\end{equation}
Hence, on each Schmidt block, Alice’s and Bob’s observables coincide, and each block supports the same operator algebra. Consequently, the block matrices of Alice and Bob coincide in the Schmidt basis, i.e. $\mathcal{A}_y=B_y$.

Block-diagonality of $\mathcal{A}_y$ and $B_y$ in the Schmidt basis leads to the following situations. If all Schmidt values are distinct, i.e., $\alpha=r$ and $m_k=1$ for every $k$, then each block is $1 \times 1$. Thus, all $\mathcal{A}_y$ ($B_y$) are diagonal in the Schmidt basis and hence mutually commuting. Since Bell nonlocality requires incompatible local measurements, this regime cannot yield a Bell violation. 

If all Schmidt values are equal, i.e., $\lambda_i=\frac{1}{r} \ \forall i$, then $D=\frac{1}{r}\openone$, and the shared state is maximally entangled state in $\mathcal{H}^r_A \otimes \mathcal{H}^r_B$. The optimality condition enforces complete anticommutation, i.e., $\qty{B_y,B_{y'}}=\qty{\mathcal{A}_y,\mathcal{A}_{y'}}=\delta_{yy'}\openone$. The minimal local Hilbert space dimension that accommodates $n$ such pairwise anticommuting Hermitian involutions is the smallest complex representation of the Clifford algebra $\mathrm{Cl}_n\qty(\mathbb{C})$, given by $r_{min}=2^{\lfloor \frac{n}{2}\rfloor}$. However, this case is sufficient, but not necessary, in general.

In general, for optimal Bell violation, degeneracy is necessary, i.e., $m_k\geq 2$ to allow non-commuting observables. Now, lets analyse the algebraic properties in each block. Since each $\mathcal{A}_y$ and $B_y$ are block diagonal, each block will only act on the corresponding block of the shared state. Thus, all optimality conditions hold independently within each block. From $\qty(\qty{\mathcal{A}_y,\mathcal{A}_{y'}}\otimes\openone) \ket{\psi} =\qty(\openone\otimes \qty{B_y,B_{y'}}) \ket{\psi}=0$, restricting to block $k$ gives
\begin{equation}
  \forall y, \ \  \qty(\qty{\mathcal{A}_y^{(k)},\mathcal{A}_{y'}^{(k)}}\otimes\openone) \ket{\phi^+} =\qty(\openone\otimes \qty{B_y^{(k)},B_{y'}^{(k)}}) \ket{\phi^+}=0.
\end{equation}
where $\ket{\phi^+}=\frac{1}{\sqrt{m_k}}\sum_{i=1}^{m_k} \ket{i}\otimes\ket{i}$ is the maximally entangled state on the $k$-th block. Using vectorisation, this becomes
\begin{equation}
\frac{1}{\sqrt{m_k}}\mathrm{vec}\qty(\qty{\mathcal{A}_y^{(k)},\mathcal{A}_{y'}^{(k)}})=\frac{1}{\sqrt{m_k}}\mathrm{vec}\qty(\qty{B_y^{(k)},B_{y'}^{(k)}}^T)=0
\end{equation}
Injectivity of the vectorisation map implies
\begin{equation}
\forall y\neq y', \ \ \qty{\mathcal{A}_y^{(k)},\mathcal{A}_{y'}^{(k)}}=0, \ \ \text{and } \qty{B_y^{(k)},B_{y'}^{(k)}}^T=0 \implies \qty{B_y^{(k)},B_{y'}^{(k)}}=0.
\end{equation}
Thus, within each Schmidt block, the local observables form a set of $n$ Hermitian involutions that pairwise anticommute.

Consequently, the degeneracy of $D$ partitions the global system into orthogonal subspaces, each supporting a maximally entangled component of $\ket{\psi}$. Within each such subspace, the observables realise a complex representation of the complex Clifford algebra $\mathrm{Cl}_n\qty(\mathbb{C})$. Hence, no block of smaller dimension $2^{\lfloor n/2 \rfloor}$ can achieve the optimal Bell violation.

Collecting all results, every optimal quantum realisation possesses the structure
\begin{equation}
    \ket{\psi}=\bigoplus_k \sqrt{\lambda^{(k)}m_k} \ket{\phi^{+}_{m_k}}, \ \ m_k \geq 2^{\lfloor \frac{n}{2} \rfloor}, \ \ \text{with } \mathcal{A}_y^{(k)}= B_y^{(k)},
\end{equation}
where each block carries an irreducible representation of the complex Clifford algebra $\mathrm{Cl}_n\qty(\mathbb{C})$, unique up to local unitary equivalence. The direct-sum freedom corresponds to classical probabilistic mixing between identical Clifford sectors, an unavoidable degeneracy that leaves the Bell functional value unchanged but does not affect the self-testing equivalence of the underlying irreducible blocks.
\section{An example for self-testing for $n=4$ case}\label{apxn4}

If a quantum strategy $\qty{A_x^{(k)},B_y^{(k)},\ket{\psi}^{(k)} \ \forall x,y \in \{1,2,3,4\}}$ provide the optimal quantum value of the Bell functional $\mathscr{G}_4$ within the $k-th$ Schmidt block, satisfying 
\begin{equation}
    \qty{\mathcal{A}_{y}^{(k)},\mathcal{A}_{y'}^{(k)}}=\qty{B_{y}^{(k)},B_{y'}^{(k)}}=2 \delta_{yy'}, \ \ \mathcal{A}_{y}^{(k)}=\qty(B_y^{(k)})^T,
\end{equation}
then there exists a unitary of the form $\mathcal{U}^{(k)} = U_A^{(k)}\otimes V_B^{(k)}$ acting such that:
\begin{equation}
    \begin{aligned}\label{A4ob}
        \qty(U_A^{(k)})^\dagger (\mathcal{A}_1) U_A^{(k)} &= \qty(\sigma_z \otimes \openone_{2})_A\otimes\openone_{J^{(k)}_A}  \ ;\  \  \
        \qty(U^{(k)}_A)^\dagger (\mathcal{A}_2) U^{(k)}_A = \qty(\sigma_y\otimes \openone_{2})_A\otimes\openone_{J^{(k)}_A},\\
        \qty(U^{(k)}_A)^\dagger (\mathcal{A}_3) U^{(k)}_A &= \qty(\sigma_x\otimes\sigma_z)_A\otimes\openone_{J^{(k)}_A}\ ;\  \ \ 
        \qty(U^{(k)}_A)^\dagger (\mathcal{A}_4) U_A^{(k)} = \qty(\sigma_x\otimes\sigma_{y})_A\otimes\openone_{J^{(k)}_A}.
    \end{aligned} 
\end{equation}
similarly for the Bob's part
\begin{equation}
    \begin{aligned}\label{B4ob}
        \qty(V_B^{(k)})^\dagger (B_1) V_B^{(k)} &= \qty(\sigma_z\otimes \openone_{2})_B\otimes\openone_{J_B^{(k)}} \ ;\  \  \
        \qty(V_B^{(k)})^\dagger (B_2) V_B^{(k)} = -\qty(\sigma_y\otimes  \openone_{2})_B\otimes\openone_{J_B^{(k)}},\\
        \qty(V_B^{(k)})^\dagger (B_3) V_B^{(k)} &=\qty(\sigma_x\otimes\sigma_z)_B\otimes\openone_{J_B^{(k)}}\ ;\  \  \
        \qty(V_B^{(k)})^\dagger (B_4) V_B^{(k)} = -\qty(\sigma_x\otimes\sigma_{y})_B\otimes\openone_{J_B^{(k)}}.
    \end{aligned}
\end{equation}
and 
\begin{equation}
    U_{A}^{(k)} \otimes V_{B}^{(k)} \ket{\psi} = \qty(\ket{\phi^+}\otimes\ket{\phi^+})_{A^\prime B^\prime}\otimes \ket{junk}_{AB}.
\end{equation}

To show such existence of local unitary, we omit the superscript $(k)$ denoting the $k$-th block. Any superscript appearing in the proof relates to iteration. To begin, without loss of generality, we consider a basis $V_1$ such that
\begin{equation}\label{b1}
    V_1 B_{1,4} V_1^\dag = \mathtt{B}_{1,4} = \sigma_z\otimes\openone_2\otimes\openone_{J_B}
\end{equation}
Moreover, let $\mathtt{B}_{2,4} =  \sum_{i,j=0}^{1} \ket{i}\bra{j}\otimes \chi_{ij}$ is in the same basis of $\mathtt{B}_{1,4}$. Then, $\mathtt{B}_{2,4}$ can be written as
\begin{equation}
    \mathtt{B}_{2,4} = \begin{bmatrix}
        \chi_{00} &\chi_{01}\\
        \chi_{01}^\dagger & \chi_{11}
    \end{bmatrix} 
\end{equation}
where $\chi_{ij}$ are hermitian matrices satisfying of dimension $2J_{B}$ and $4{J_B} = d$. Now, using the anticommutation relation $\{\mathtt{B}_1,\mathtt{B}_2\}=0$, we get $\mathtt{B}_{2,4} = -\mathtt{B}_{1,4}\mathtt{B}_{2,4}\mathtt{B}_{1,4}$. Putting Eq. (\ref{b1}) in the above expression for $\mathtt{B}_{2,4}$, we obtain
\begin{equation}\label{b21}
    \mathtt{B}_{2,4} = \begin{bmatrix}
        0 & X_{2,4}\\
        X_{2,4}^\dag & 0
    \end{bmatrix} \ \ \ \
\end{equation}
where $X_{2,4}=\chi_{01}$.  
Similarly, we get
\begin{equation}
    \begin{aligned}
        \mathtt{B}_{3,4} &= \begin{bmatrix}
            0 & X_{3,4}\\
            X_{3,4}^\dagger &0
        \end{bmatrix} \ \ \ \;
        \mathtt{B}_{4,4} = \begin{bmatrix}
            0 & X_{4,4}\\
            X_{4,4}^\dagger &0
        \end{bmatrix}.
    \end{aligned}
\end{equation}
Let us construct an unitary of the form 
\begin{equation}
    V_2 = \begin{bmatrix}
        \openone_{2J_B} & 0\\
        0 & -\iota X_{2,4}
    \end{bmatrix}
\end{equation}
This operates on $\mathtt{B}_{1,4}$ as $V_2( \sigma_z\otimes\openone_2\otimes\openone_{J_B})V_2^\dagger = \sigma_z\otimes\openone_2\otimes\openone_{J_B}$. Acting the same unitary on $\mathtt{B}_{2,4}$, gives
\begin{equation*}
     \begin{aligned}
        \mathtt{B}_2 &=V_2 \begin{bmatrix}
            0 & X_{2,4}\\
            X_{2,4}^\dagger & 0
        \end{bmatrix}V_2^\dagger = \begin{bmatrix}
            0 & i\openone_{2J_B}\\
            -i\openone_{2J_B}^\dagger & 0
        \end{bmatrix} = -\sigma_y\otimes\openone_2\otimes\openone_{J_B}.
    \end{aligned}
\end{equation*}
Similarly, considering $\{\mathtt{B}_{2,4},\mathtt{B}_{3,4}\} =0$, we get 
\begin{equation}\label{b22}
     \begin{aligned}\mathtt{B}_{3,4}  =
\begin{bmatrix}
          0 & X_{3,4}X_{2,4}^\dagger X_{3,4}\\
            X_{3,4}^\dagger X_{2,4}  X_{3,4}^\dagger & 0
        \end{bmatrix}.
         \end{aligned}
        \end{equation}
        Comparing Eqs (\ref{b21}) and (\ref{b22}), we get $X_{2,4} = -X_{3,4}X_{2,4}^\dagger X_{3,4}$. Using this relation in the expression of $\mathtt{B}_{3,4}$ after the application of unitary operation $V$,
\begin{equation}
     \begin{aligned}
        \mathtt{B}_{3,4} &=V_2 \begin{bmatrix}
            0 & X_{3,4}\\
            X_{3,4}^\dagger & 0
        \end{bmatrix}V_2^\dagger = \begin{bmatrix}
            0 & \iota X_{3,4}X_{2,4}^\dagger\\
            -\iota X_{2,4}X_{3,4}^\dagger & 0
        \end{bmatrix} = \sigma_x\otimes(\iota X_{3,4}X_{2,4}^\dagger).
    \end{aligned}
\end{equation}
Similar calculation yields $\mathtt{B}_{4,4} = \sigma_x\otimes(\iota X_{4,4}X_{2,4}^\dagger)$. Note here that $\iota X_{3,4}X_{2,4}^\dagger$ and $\iota X_{4,4}X_{2,4}^\dagger$ are hermitian unitary operators. Interestingly, $\mathtt{B}_{3,4}$ and $\mathtt{B}_{4,4}$ satisfy $\{\mathtt{B}_{3,4},\mathtt{B}_{4,4}\}=0$ implying $
\iota X_{3,4}X_{2,4}^\dagger$ and $\iota X_{4,4}X_{2,4}^\dagger$ needs to be anti-commuting. To obtain the necessary simplifications, we have to construct $V_1^{(1)}$ such that $V_1^{(1)} \iota X_{3,4}X_{2,4}^\dagger \qty(V_1^{(1)})^\dag= \sigma_z\otimes\openone_{J_B}$ and the anticommuativity of $
\iota X_{3,4}X_{2,4}^\dagger$ and $\iota X_{4,4}X_{2,4}^\dagger$ gives 
\begin{equation}
    V_1^{(1)} \qty[\iota X_{4,4}X_{2,4}^\dagger] \qty(V_1^{(1)})^\dag= \begin{bmatrix}
        0 & X_4^{(1)}\\
        \qty(X_4^{(1)})^\dag & 0
    \end{bmatrix}
\end{equation}
Now we can define a unitary matrix 
\begin{equation}
    V_2^{(1)} = \begin{bmatrix}
        \openone_{J_B} & 0\\
        0 & -\iota X_4^{(1)}
    \end{bmatrix}
\end{equation}
which gives
\begin{equation}
   V_2^{(1)} V_1^{(1)} \qty[\iota X_{4,4}X_{2,4}^\dagger] \qty(V_1^{(1)})^\dag \qty(V_2^{(1)})^\dag= -\sigma_y\otimes\openone_{J_A}
\end{equation}
Hence, considering the unitary $V_B = \qty(\openone_2\otimes V_2^{(1)} V_1^{(1)} )V_2V_1$ gives the desired result
\begin{equation}
    \begin{aligned}
        V_B B_{1,4} V_B & = \sigma_z\otimes\openone_2\otimes\openone_{J_B}\\
        V_B B_{2,4} V_B & = -\sigma_y\otimes\openone_2\otimes\openone_{J_B}\\
        V_B B_{3,4} V_B & = \sigma_x\otimes\sigma_z\otimes\openone_{J_B}\\
        V_B B_{4,4} V_B & = -\sigma_x\otimes\sigma_y\otimes\openone_{J_B}
    \end{aligned}
\end{equation}
which satisfies given Eq. (\ref{B4ob}). For Alice's observables, we are dropping the notation $\mathcal{A}_{y,n}$ and will just use $\mathcal{A}_{y}$ for the proof to be simplified. Similar to the Bob's case, we can construct a unitary $U_1$ such that $U_1 \mathcal{A}_1 U_1^\dag = \sigma_z\otimes\openone_2\otimes\openone_{J_A}$ Then, $U_1 \mathcal{A}_2 U_1^\dag$ can be written as
\begin{equation}
    U_1 \mathcal{A}_2 U_1^\dag= \begin{bmatrix}
        \eta_{00} &\eta_{01}\\
        \eta_{01}^\dagger & \eta_{11}
    \end{bmatrix} 
\end{equation}
where $\eta_{ij}$ are hermitian matrices satisfying of dimension $2J_{A}$ and $4{J_A} = d$. 
\begin{equation}\label{a21}
     U_1 \mathcal{A}_2 U_1^\dag = \begin{bmatrix}
        0 & Y_2\\
        Y_2^\dag & 0
    \end{bmatrix} \ \ \ \
\end{equation}
where $Y_2=\eta_{01}$.  
Similarly, we get
\begin{equation}
    \begin{aligned}
        U_1 \mathcal{B}_3 U_1^\dag &= \begin{bmatrix}
            0 & Y_3\\
            Y_3^\dagger &0
        \end{bmatrix} \ \ \ \
        U_1 \mathcal{B}_4 U_1^\dag = \begin{bmatrix}
            0 & Y_4\\
            Y_4^\dagger &0
        \end{bmatrix}
    \end{aligned}
\end{equation}
 We construct a unitary  of the form 
\begin{equation}
    U_2 = \begin{bmatrix}
        \openone_{2J_A} & 0\\
        0 & \iota Y_2
    \end{bmatrix}
\end{equation}
This gives $U_2U_1 \mathcal{A}_2 U_1^\dag U_2^\dag = \sigma_y\otimes\openone_2\otimes\openone_{J_A}$. We will follow the similar procedure as done for the Bob's observables, the only difference being that we construct 
\begin{equation}
    U_2^{(1)} = \begin{bmatrix}
        \openone_{J_A} & 0\\
        0 & \iota Y_4^{(1)}
    \end{bmatrix}
\end{equation}
This gives us the final unitary $U_A =  \qty(\openone_2 \otimes U_2^{(1)} U_1^{(1)})U_2U_1$ which gives us
\begin{equation}
    \begin{aligned}
        U_A \mathcal{A}_1 U_A & = \sigma_z\otimes\openone_2\otimes\openone_{J_A}\\
        U_A \mathcal{A}_2 U_A & = \sigma_y\otimes\openone_2\otimes\openone_{J_A}\\
        U_A \mathcal{A}_3 U_A & = \sigma_x\otimes\sigma_z\otimes\openone_{J_A}\\
        U_A \mathcal{A}_4 U_A & = \sigma_x\otimes\sigma_y\otimes\openone_{J_A}
    \end{aligned}
\end{equation}
which satisfies Eq. (\ref{A4ob}). Now, to demonstrate that the state $\ket{\psi}$ of the physical system in the basis $\mathcal{U} = U_A\otimes V_B$ is equivalent to two copies of maximally entangled state in the reference system, we start with self-testing relations $\mathcal{A}_y\otimes\mathcal{B}_y\ket{\psi} = \ket{\psi}$. Since the state  $\ket{\psi} \in \mathcal{H}_A \otimes \mathcal{H}_B \subseteq \mathcal{H}_{A^\prime} \otimes \mathcal{H}_{B^\prime} \otimes \mathcal{H}_{J_A} \otimes \mathcal{H}_{J_B}$ we can write the state $\mathcal{U}\ket{\psi}$ as
\begin{equation}\label{4state}
     \ket{\xi}=\mathcal{U}\ket{\psi}=\sum_{i_{1}i_{2}j_{1}j_{2}=0}^{1}\ket{i_{1}i_{2}j_{1}j_{2}}\otimes \ket{\psi_{i_1j_1i_2j_2}}
\end{equation}
 Note that the self-testing relation $\mathcal{A}_y\otimes B_y \ket{\psi} = \ket{\psi}$ is also satisfied in the the rotated frame
 \begin{equation}\label{apxobsrel}
      \qty(\mathcal{U} \mathcal{A}_y\otimes B_y\mathcal{U}^\dag) \ket{\xi} = \ket{\xi} 
 \end{equation}
 Now, considering the relation Eq.~(\ref{apxobsrel}) for $y=1$  and rearranging Hilbert spaces appropriately and, we get
 \begin{equation}
     \begin{aligned}
\qty[\qty(\sigma_z\otimes \openone\otimes\sigma_z\otimes\openone)\otimes \openone_{J_A}\otimes\openone_{J_B}] \sum_{i_{1}i_{2}j_{1}j_{2}=0}^{1}\ket{i_{1}i_{2}j_{1}j_{2}}\otimes \ket{\psi_{i_1j_1i_2j_2}} = \sum_{i_{1}i_{2}j_{1}j_{2}=0}^{1}\ket{i_{1}i_{2}j_{1}j_{2}}\otimes \ket{\psi_{i_1j_1i_2j_2}}\\
     \end{aligned}
 \end{equation}
Which gives us the following equation
\begin{equation}
    (-1)^{i_{1}+j_{1}}\sum_{i_{1}i_{2}j_{1}j_{2}=0}^{1}\ket{i_{1}i_{2}j_{1}j_{2}}\otimes \ket{\psi_{i_1j_1i_2j_2}} = \sum_{i_{1}i_{2}j_{1}j_{2}=0}^{1}\ket{i_{1}i_{2}j_{1}j_{2}}\otimes \ket{\psi_{i_1j_1i_2j_2}}
\end{equation}
Upon solving the above equation, we obtain 
\begin{eqnarray}
 (-1)^{i_{1}+j_{1}} \ket{\psi_{i_1j_1i_2j_2}}=\ket{\psi_{i_1j_1i_2j_2}}
\end{eqnarray}
which provides $\ket{\psi_{i_1,i_2,j_1,j_2}} = 0$ whenever  $(i_1+j_1)\ mod \ 2 = 1$\textbf{} implying $\ket{\psi_{0010}} = \ket{\psi_{0011}}=\ket{\psi_{0110}}=\ket{\psi_{0111}}=\ket{\psi_{1000}}=\ket{\psi_{1001}}=\ket{\psi_{1100}}=\ket{\psi_{1101}}=0$. Hence, putting the above values in Eq.~(\ref{4state}) we get
\begin{equation}\label{4psin}
\begin{aligned}
     \ket{\xi} &= \ket{0000}\otimes\ket{\psi_{0000}}+\ket{0001}\otimes\ket{\psi_{0001}}+\ket{0100}\otimes\ket{\psi_{0100}}+\ket{0101}\otimes\ket{\psi_{0101}}\\
     &+\ket{1010}\otimes\ket{\psi_{1010}}+\ket{1011}\otimes\ket{\psi_{1011}}+\ket{1110}\otimes\ket{\psi_{1110}}+\ket{1111}\otimes\ket{\psi_{1111}}
\end{aligned}
\end{equation}
Now using $\ket{\psi}$ in the relation Eq.~(\ref{apxobsrel}) for $y=2$ we can obatin
\begin{equation}
    \sum_{i_{1}i_{2}j_{1}j_{2}=0}^{1}\ket{i_{1}i_{2}j_{1}j_{2}}\otimes \ket{\psi_{i_1j_1i_2j_2}}=\sum_{i_{1}i_{2}j_{1}j_{2}=0}^{1}\ket{i_{1}i_{2}j_{1}j_{2}}\otimes \ket{\psi_{i_1j_1i_2j_2}}
\end{equation}
and solving we get $\ket{\psi_{i_1,i_2,j_1,j_2}}=\ket{\psi_{i_1',i_2',j_1',j_2'}}$ when $i_2 = i_2', j_2=j_2'$, which further gives, $\ket{\psi_{1010}}=\ket{\psi_{0000}}, \ket{\psi_{1011}}=\ket{\psi_{0001}}, \ket{\psi_{0100}}=\ket{\psi_{1110}}, \ket{\psi_{0101}}=\ket{\psi_{1111}}$. Now substituting the values in Eq. (\ref{4psin}), we have
\begin{equation}\label{4psinnn}
\begin{aligned}
     \ket{\xi} &= (\ket{0000}+\ket{1010})\otimes\ket{\psi_{0000}}+(\ket{0001}+\ket{1011})\otimes\ket{\psi_{0001}}\\
     &+(\ket{0100}+\ket{1110})\otimes\ket{\psi_{1110}}+(\ket{0101}+\ket{1111})\otimes\ket{\psi_{1111}}
\end{aligned}
\end{equation}
Next, using  Eq.~(\ref{apxobsrel}) for $y=3$ and simplifying, we get $\ket{\psi_{i_1,i_2,j_1,j_2}}=0$ when $i_1=i_2$ and $j_1\neq j_2$ which gives $\ket{\psi_{0001}}=\ket{\psi_{1110}}=0$. Again substituting the values in Eq.~(\ref{4psinnn})
\begin{equation}
\begin{aligned}
     \ket{\xi} &= (\ket{0000}+\ket{1010})\otimes\ket{\psi_{0000}}+(\ket{0101}+\ket{1111})\otimes\ket{\psi_{1111}}
\end{aligned}
\end{equation}
Using the relation Eq.~(\ref{apxobsrel}) for $y=4$, we obtain $\ket{\psi_{1111}}=\ket{\psi_{0000}}$. Thus, the final state with appropriate normalization becomes
\begin{equation}
    \ket{\xi}=\frac{1}{2}\qty(\ket{0000}+\ket{0101}+\ket{1010}+\ket{1111})\otimes2\ket{\psi_{0000}}
\end{equation}
where the first two bits are with Alice and the last two are with Bob. Rearranging the reference systemsystem, we obtain
\begin{equation}
    \ket{\xi} = \qty(\ket{\phi^+}_{i_1j_1}\otimes\ket{\phi^+}_{i_2j_2})\otimes2\ket{\psi_{0000}}
\end{equation}
Hence, the unitary $\mathcal{U}$ extracts two copies of maximally entangled states in the reference system.


\section{Bounding the distance of the physical State from the ideal state: } \label{statebound}

For the ideal case, the SOS conditions imply that the state lies entirely in the intersection of the kernels of these operators. Any physical state achieving a slightly suboptimal violation necessarily has a component outside this ideal subspace, and the magnitude of this component provides a direct measure of the deviation from the ideal state. By explicitly projecting the physical state onto the ideal subspace, we obtain quantitative bounds on the norm of the orthogonal component, thereby providing a precise characterization of state robustness.

From SOS, the relation between the ideal observables and the state
\begin{equation}\label{apidsos}
  \forall x \ \ \   M_x\ket{\psi}=0 \implies A_x \otimes\openone \ket{\psi}=\frac{1}{\sqrt{n}}(\openone\otimes\mathcal{B}_x)\ket{\psi},
\end{equation}
The ideal SOS relation $M_x\ket{\psi}=0$ implies that the state $\ket{\psi}$ is in the kernel of $M_x$. The kernel is defined as
\begin{equation}
    \ker(M_x)=\qty{\ket{\psi}:M_x\ket{\psi}=0}
\end{equation}
representing the set of all states that satisfy the ideal SOS relation perfectly. As there are $2^{n-1}$ such condition, we define the subspace of all states satisfying every constraint simultaneously as the intersection of these kernels
\begin{equation}
    \mathcal{S}=\bigcap\limits_x  \ker(M_x)
\end{equation}
The projection of a physical state $\ket{\tilde{\psi}}$ onto this subspace yields the closest ideal state, denoted $\ket{\psi'}$
\begin{equation}
    \ket{\psi'}=\frac{\Pi \ket{\tilde{\psi}}}{\norm{\Pi \ket{\tilde{\psi}}}}, \ \ \ {\norm{\Pi \ket{\tilde{\psi}}}}=\Tr[\Pi \dyad{\tilde{\psi}}]
\end{equation}
where, $\Pi$ is the orthogonal projector onto $\mathcal{S}$. Decomposing $\ket{\tilde{\psi}}$ into components parallel and orthogonal to the ideal subspace, one has
\begin{equation}\label{orpros}
    \ket{\tilde{\psi}}=\Pi\ket{\tilde{\psi}} +\qty(\openone-\Pi)\ket{\tilde{\psi}}
\end{equation}
Here, $\Pi\ket{\tilde{\psi}}$ lies within the ideal subspace, while $\qty(\openone-\Pi)\ket{\tilde{\psi}}$ is the component orthogonal to it. The norm of the latter component, $\norm{\qty(\openone-\Pi)\ket{\tilde{\psi}}}$, quantifies the distance between $\ket{\tilde{\psi}}$ and $\mathcal{S}$. Defining
\begin{equation}
    \ket{\psi_{\parallel}}:=\frac{\Pi\ket{\tilde{\psi}}}{\norm{\Pi\ket{\tilde{\psi}}}}, \ \ \text{and } \ \ket{\psi_{\perp}}:= \frac{\qty(\openone-\Pi)\ket{\tilde{\psi}}}{\norm{\qty(\openone-\Pi)\ket{\tilde{\psi}}}}
\end{equation}
we note that $\ket{\psi_{\parallel}}$ and $\ket{\psi_{\perp}}$ are orthogonal, i.e., $\ip{ \psi_{\parallel}}{ \psi_{\perp}}=0$, and can express $\ket{\tilde{\psi}}$ as
\begin{equation}
    \ket{\tilde{\psi}} = \sqrt{1-\epsilon^2} \ket{\psi_{\parallel}}+ \epsilon \ket{\psi_{\perp}}, \ \ \ \text{where } \epsilon=\norm{\qty(\openone-\Pi)\ket{\tilde{\psi}}}
\end{equation}
The distance between the ideal state $(\ket{\psi}\equiv\ket{\psi'}\equiv\ket{\psi_{\parallel}})$ and the physical state $\ket{\tilde{\psi}}$ then becomes
\begin{equation}
\begin{aligned}
        \norm{\ket{\tilde{\psi}}-\ket{\psi'}}^2&=\norm{\sqrt{1-\epsilon^2} \ket{\psi_{\parallel}}+ \epsilon \ket{\psi_{\perp}}-\ket{\psi_{\parallel}}}^2\\
        &=\norm{\qty(\sqrt{1-\epsilon^2}-1) \ket{\psi_{\parallel}}+ \epsilon \ket{\psi_{\perp}}}^2\\
        &= \qty(\sqrt{1-\epsilon^2}-1)^2+\epsilon^2 \ \ \text{[since $\ket{\psi_{\parallel}}$ and $\ket{\psi_{\perp}}$ are normalised and orthogonal]}\\
        &=2\qty(1-\sqrt{1-\epsilon^2}) \approx \epsilon^2 =\norm{\qty(\openone-\Pi)\ket{\tilde{\psi}}}^2
\end{aligned}
\end{equation}
Recalling the positive semidefinite operator from Eq.~(\ref{sstart}),
\begin{equation}
    \mathcal{M}=\sum_x \frac{\omega_x}{2}M_x^{\dagger}M_x, \ \ \ \text{with } \delta=\expval{\mathcal{M}}{\tilde{\psi}}
\end{equation}
we observe that $M_x\ket{\psi}=0 \ \forall x$, implies $\mathcal{M}\ket{\psi}=0$, so that the kernel of $\mathcal{M}$ coincides with $\mathcal{S}$
\begin{equation}
    \ker(\mathcal{M})=\mathcal{S}
\end{equation}
That means the projection onto $\ker(\mathcal{M})$ is exactly $\Pi$, i.e. $\mathcal{M}\Pi=0$. Thus, from Eq.~(\ref{orpros}), we get
\begin{equation}\label{expM}
\begin{aligned}
        \mathcal{M} \ket{\tilde{\psi}}&=\mathcal{M} \Pi \ket{\tilde{\psi}}+\mathcal{M}\qty(\openone-\Pi)\ket{\tilde{\psi}}\\
        &= \mathcal{M}\qty(\openone-\Pi)\ket{\tilde{\psi}}\\
        \implies & \expval{\mathcal{M}}{\tilde{\psi}}=\norm{\qty(\openone-\Pi)\ket{\tilde{\psi}}}^2
\end{aligned}
\end{equation}
Since $\mathcal{M}$ is a positive semidefinite operator acting as zero on $\mathcal{S}$, its nonzero eigenvalues correspond to directions orthogonal to $\mathcal{S}$. Denoting the spectral decomposition of $\mathcal{M}$ as
\begin{equation}
    \mathcal{M}=\sum_i \lambda_i \dyad{v_i}, \ \ \lambda_i\geq 0
\end{equation}
where $\qty{\ket{v_i}}$ span the full space and $\lambda_j=0$ for all $\ket{v_j}\in \mathcal{S}$. The smallest nonzero eigenvalue on the orthogonal complement $\mathcal{S}^{\perp}$
\begin{equation}\label{lminc1}
\begin{aligned}
       \lambda_{min}&=\min\limits_{\ket{\phi}\in \mathcal{S}^{\perp},\norm{\ket{\phi}}=1} \expval{\mathcal{M}}{\phi}\\
       &=\min\limits_{\ket{\phi}\in \mathcal{S}^{\perp},\norm{\ket{\phi}}=1} \expval{\sum_x \frac{\omega_x}{2}M_x^{\dagger}M_x}{\phi}\\
       &=\min\limits_{\ket{\phi}\in \mathcal{S}^{\perp},\norm{\ket{\phi}}=1} \sum_x \frac{\omega_x}{2}\expval{M_x^{\dagger}M_x}{\phi} \\
\end{aligned}
\end{equation}
Each $M_x$ is constructed from Hermitian and unitary components, $A_x\otimes\openone$ and $\openone\otimes \frac{\mathcal{B}_x}{\omega_x}$. Consequently, the eigenvalues of $M_x$ are restricted to the set $\qty{-2,0,2}$. Thus, $M_x^{\dagger}M_x$ has eigenvalues $\qty{0,4}$, with $0$ eigenvalue corresponds to the ideal subspace $\mathcal{S}$, i.e., the set of states that satisfy $M_x\ket{\psi}=0$. This ensures that the ideal state lies entirely within $\mathcal{S}$, and any component of a physical state outside this subspace corresponds to a positive eigenvalue of $M_x^{\dagger}M_x$.

To identify the smallest nonzero eigenvalue $\lambda_{min}$ of $\mathcal{M}$, we consider the spectral decomposition on the orthogonal complement $\mathcal{S}^{\perp}$. On this subspace, each $M_x^{\dagger}M_x$ acts positively, and since the zero eigenvalues are entirely contained within $\mathcal{S}$,the minimal nonzero contribution comes from the direction in $\mathcal{S}^{\perp}$ that aligns with only one of the $M_x^{\dagger}M_x$ terms being nonzero while the rest vanish. Denoting this particular operator as $M_j$, we can write
\begin{equation}\label{lminc2}
\begin{aligned}
       \lambda_{min}&\geq \min\limits_{\ket{\phi'}\in \mathcal{S},\norm{\ket{\phi'}}=1}  \sum_{x\neq j} \frac{\omega_x}{2}\expval{M_x^{\dagger}M_x}{\phi} + \frac{\omega_j}{2}\expval{M_j^{\dagger}M_j}{\phi^{*}}, \ \ \ket{\phi^{*}}\in \mathcal{S}^{\perp}  \\
       &=2 \sqrt{n}
\end{aligned}
\end{equation}
where $\ket{\phi^{*}}\in \mathcal{S}^{\perp}$ is the eigenvector associated with the minimal nonzero eigenvalue. Then, for any state, from Eq.~(\ref{expM}), we obtain
\begin{equation}
    \expval{\mathcal{M}}{\tilde{\psi}}\geq \lambda_{min} \norm{\qty(\openone-\Pi)\ket{\tilde{\psi}}}^2
\end{equation}
Thus,
\begin{equation}
  \norm{\ket{\tilde{\psi}}-\ket{\psi'}}^2 \approx  \norm{\qty(\openone-\Pi)\ket{\tilde{\psi}}}^2 \leq \frac{\norm{\mathcal{M}\ket{\tilde{\psi}}}^2}{\lambda_{min}} = \frac{1}{\lambda_{min}}\sum_{x = 1}^{2^{n-1}} \frac{\omega_x}{2} \norm{M_x\ket{\tilde{\psi}}}^2=\frac{\delta}{\lambda_{min}}\sim \frac{\delta}{2 \sqrt{n}}  \ \ [\text{from Eq.~(\ref{lminc2})}]
\end{equation}
Therefore,
\begin{equation}\label{suppstaterof}
  \norm{\ket{\tilde{\psi}}-\ket{\psi'}} \sim \sqrt{\frac{\delta}{2 \sqrt{n}}} 
\end{equation}


\section{Robust Bound on Bob’s Anticommutation Relations} \label{roancobb}

Recalling Eq.~(88) from the main text,
\begin{equation}\label{apomro}
\tilde{\omega}_x^2=\expval{\tilde{\mathcal{B}}_x^2}_{\tilde{\psi}}=n+\sum_{y<y'} (-1)^{z^x_y+z^x_{y'}} \expval{\tilde{B}_y,\tilde{B}_{y'}}_{\tilde{\psi}} = n+\Delta_x \ \ \ \forall x
\end{equation}
Taking the square root and applying a Taylor expansion yields
\begin{equation}
    \tilde{\omega}_x=\sqrt{n}\qty(1+\frac{\Delta_x}{n})^{\frac{1}{2}}\approx \sqrt{n} \qty(1+\frac{\Delta_x}{2n}-\frac{\Delta_x^2}{8n^2}+\order{\frac{\Delta_x^3}{n^3}})
\end{equation}
Since $\Delta_x \ll n$, implying $\eval{\order{\frac{\Delta_x^m}{n^m}}}_{m\geq 2}\approx 0$, we obtain
\begin{equation}\label{tildeomega}
    \tilde{\omega}_x \approx \sqrt{n} \qty(1+\frac{\Delta_x}{2n})
\end{equation}
Let us now consider the operator $\tilde{M}_x=\frac{1}{\tilde{\omega}_x}(\openone\otimes\tilde{\mathcal{B}}_x)-\tilde{A}_x \otimes \openone$. Then,
\begin{equation}\label{mbound1}
      \norm{\tilde{M}_x \ket{\tilde{\psi}}}^2=\qty(1+\frac{\expval{\tilde{\mathcal{B}}_x^2}}{\tilde{\omega}_x^2}-\frac{2}{\tilde{\omega}_x}\expval{\tilde{A}_x\otimes \tilde{\mathcal{B}}_x}) =2\qty[1-\frac{1}{\sqrt{n}}\qty(1-\frac{\Delta_x}{2n})\expval{\tilde{A}_x\otimes \tilde{\mathcal{B}}_x}]
\end{equation}
Applying the Cauchy–Schwarz inequality to $\expval{\tilde{A}_x\otimes \tilde{\mathcal{B}}_x}$, we have
\begin{equation}\label{expb}
    \abs{\expval{\tilde{A}_x\otimes \tilde{\mathcal{B}}_x}} \leq \sqrt{\expval{\tilde{A}_x^2}\expval{\tilde{\mathcal{B}}_x^2}}=\sqrt{n+\Delta_x}\approx \sqrt{n}\qty(1+\frac{\Delta_x}{2n})
\end{equation}
Substituting this into the previous Eq.~(\ref{mbound1}) gives 
\begin{equation}
   \norm{\tilde{M}_x \ket{\tilde{\psi}}}^2 \geq \frac{\Delta_x^2}{4n^2} 
\end{equation}
Summing over all $x$ and invoking the SOS relation yields
\begin{equation}
     \delta = \sum_{x = 1}^{2^{n-1}} \frac{\omega_x}{2} \norm{\tilde{M}_x\ket{\tilde{\psi}}}^2 \geq \sum_{x = 1}^{2^{n-1}} \frac{1}{2} \qty(1+\frac{\Delta_x}{2n})\frac{\Delta_x^2}{4n^2}\approx \sum_{x = 1}^{2^{n-1}} \frac{\Delta_x^2}{8n^2}
\end{equation}
Hence, the deviation parameter $\Delta_x$ satisfies the bound
\begin{equation}\label{newD}
    \abs{\Delta_x} \leq \frac{1}{2^{n+1}n}\sqrt{\delta} 
\end{equation}

Now, recall the matrix $C \in \mathbb{R}^{2^{n-1} \times \binom{n}{2}}$ defined in Eq.~(\ref{landc1}),
\begin{equation}
    C_{x,(y,y')}=(-1)^{z^x_y+z^x_{y'}}
\end{equation}
where the columns are indexed by the pairs $(y,y')$ with $1 \leq y < y' \leq n$. As proven in Appx. \ref{appacp}, this matrix has full column rank, i.e., $\text{rank}(C) = \binom{n}{2}$. Eq.~(\ref{apomro}) can thus be written in matrix form as
\begin{equation}
    \Delta=C \ \eta,
\end{equation}
where $\eta=\qty(\eta_{yy'})_{y<y'}$ is the vector of pairwise anticommutator expectations. Since $C$ has full column rank, the associated linear map $T: \mathbb{R}^{\binom{n}{2}} \to \mathbb{R}^{2^{n-1}}$ defined by $T(\eta) = C\eta$ is injective. Because the columns are independent, i.e., the only solution to $C\eta=0$ is $\eta=0$, i.e., no two distinct inputs map to the same output. Therefore, the pseudoinverse $C^+$ exists and satisfies $C^+C = I_{\binom{n}{2}}$. Therefore, $\eta=C^+\Delta$.

To bound the entries of $\eta$ in terms of $\Delta$,  we employ the $\ell_\infty$-norm (maximum absolute entry), defined as $\norm{v}_{\infty}:=\max_i |v_i|$. Then
\begin{equation} \label{normin}
    \norm{\eta}_{\infty}=\norm{C^+\Delta}_{\infty}\leq \norm{C^+}_{\infty\to \infty}\norm{\Delta}_{\infty}
\end{equation}
Here $\norm{C^+}_{\infty\to \infty}$ is the operator norm of $C^+$ as a map from $\ell_\infty$ to $\ell_\infty$, i.e., the maximum value of the $\ell_\infty$ norm under $C^+$. Component-wise, this implies
\begin{equation}
    \max\limits_{y<y'} |\eta_{yy'}|\leq \norm{C^+}_{\infty\to\infty} \max_x |\Delta_x|
\end{equation}
This tells us the largest entry of $\eta$ is controlled by the largest entry of $\Delta$, scaled by the norm of the pseudoinverse. Defining $K_n:=\norm{C^+}_{\infty\to\infty}>0$, we obtain the compact bound
\begin{equation} \label{etab}
    \max\limits_{y<y'}|\eta_{yy'}|\leq K_n \max_x |\Delta_x| 
\end{equation}
Finally, using the previously derived bound for $\Delta_X$ in Eq.~(\ref{newD}), we arrive at
\begin{equation} \label{acbobf}
    \max\limits_{y<y'}|\eta_{yy'}|\leq L_n \sqrt{\delta}, \ \ \  L_n=\frac{1}{2^{n+1}n} K_n
\end{equation}


\section{Derivation of Robustness Bounds for Measurement Observables: } \label{roobs}

We begin by inserting and subtracting the correlated term involving $\tilde{\mathcal{B}}_x$
\begin{equation}
    \qty(\tilde{A}_x-A_x)\otimes \openone=\qty(\tilde{A}_x\otimes \openone-\frac{\openone\otimes \tilde{\mathcal{B}}_x}{\tilde{\omega}_x})+\qty(\frac{\openone\otimes \tilde{\mathcal{B}}_x}{\tilde{\omega}_x}-A_x\otimes \openone).
\end{equation}
Physically, this separates Alice’s deviation into a part correlated with Bob’s operators and a residual part that accounts for the difference between the ideal and physical correlations. Taking the norm acting on $\ket{\tilde{\psi}}$ and applying the triangle inequality, $\norm{P+Q}\leq \norm{P}+\norm{Q}$, yields
\begin{equation}\label{arob}
\begin{aligned}
     \norm{\qty(\tilde{A}_x-A_x)\otimes \openone \ket{\tilde{\psi}}}&\leq\norm{\qty(\tilde{A}_x\otimes \openone-\frac{\openone\otimes \tilde{\mathcal{B}}_x}{\tilde{\omega}_x})\ket{\tilde{\psi}}}+\norm{\qty(\frac{\openone\otimes \tilde{\mathcal{B}}_x}{\tilde{\omega}_x}-A_x\otimes \openone)\ket{\tilde{\psi}}}\\
     &\leq F_n \sqrt{\delta} +\norm{\qty(\frac{\openone\otimes \tilde{\mathcal{B}}_x}{\tilde{\omega}_x}-A_x\otimes \openone)\ket{\tilde{\psi}}}  \\
     &=F_n \sqrt{\delta} +\norm{\qty{\frac{\openone\otimes \qty(\tilde{\mathcal{B}}_x-\mathcal{B}_x)}{\tilde{\omega}_x}+\qty(\frac{\openone\otimes\mathcal{B}_x}{\tilde{\omega}_x}-A_x\otimes \openone)}\ket{\tilde{\psi}}}\\
     &\leq F_n \sqrt{\delta} +\norm{\frac{\openone\otimes \qty(\tilde{\mathcal{B}}_x-\mathcal{B}_x)}{\tilde{\omega}_x}\ket{\tilde{\psi}}}+\norm{\qty(\frac{\openone\otimes\mathcal{B}_x}{\tilde{\omega}_x}-A_x\otimes \openone)\ket{\tilde{\psi}}}.
     \end{aligned}
     \end{equation}
Now, consider the second term $\norm{\qty(\frac{\openone\otimes\mathcal{B}_x}{\tilde{\omega}_x}-A_x\otimes \openone)\ket{\tilde{\psi}}}$. Using the Cauchy–Schwarz inequality and simple algebra, we obtain
\begin{equation}\label{obsrobustmid}
    \begin{aligned}
        \norm{\qty(\frac{\openone\otimes\mathcal{B}_x}{\tilde{\omega}_x}-A_x\otimes \openone)\ket{\tilde{\psi}}}&=\norm{\frac{\openone\otimes\mathcal{B}_x}{\tilde{\omega}_x}\qty(\ket{\tilde{\psi}}-\ket{\psi})+A_x\otimes\openone\qty(\ket{\tilde{\psi}}-\ket{\psi})+ \qty(\frac{\openone\otimes\mathcal{B}_x}{\tilde{\omega}_x}-A_x\otimes \openone) \ket{\psi}} \\
        &\leq \norm{\frac{\openone\otimes\mathcal{B}_x}{\tilde{\omega}_x}\qty(\ket{\tilde{\psi}}-\ket{\psi})}+ \norm{A_x\otimes\openone\qty(\ket{\tilde{\psi}}-\ket{\psi})}+\norm{ \qty(\frac{\openone\otimes\mathcal{B}_x}{\tilde{\omega}_x}-A_x\otimes \openone) \ket{\psi}}\\
        &\leq \sqrt{n}\qty(\frac{1}{\tilde{\omega}_x}+1)\norm{\ket{\tilde{\psi}}-\ket{\psi}}+ \norm{ \qty(\frac{\openone\otimes\mathcal{B}_x}{\tilde{\omega}_x}-A_x\otimes \openone) \ket{\psi}} \ \ \ \ [\text{since }\norm{\mathcal{B}_x\ket{\psi}}=\omega_x=\sqrt{n}].
    \end{aligned}
\end{equation}
Using the ideal SOS condition given by Eq.~(\ref{aborf}), we reduce the last term $(A_x\otimes\openone)\ket{\psi}=\frac{1}{\sqrt{n}}\qty(\openone\otimes\mathcal{B}_x)\ket{\psi}$ of Eq.~(\ref{obsrobustmid}) to
\begin{equation}\label{obsrobustmid1}
    \begin{aligned}
    \norm{\qty(\frac{\openone\otimes\mathcal{B}_x}{\tilde{\omega}_x}-A_x\otimes \openone)\ket{\tilde{\psi}}} &\leq \sqrt{n}\qty(\frac{1}{\tilde{\omega}_x}+1)\norm{\ket{\tilde{\psi}}-\ket{\psi}}+ \norm{ \qty(\frac{\openone\otimes\mathcal{B}_x}{\tilde{\omega}_x}-\frac{\openone\otimes\mathcal{B}_x}{\sqrt{n}}) \ket{\psi}}  \\
    &\leq\sqrt{n}\qty(\frac{1}{\tilde{\omega}_x}+1)C_n \sqrt{\delta}+ \abs{\qty(\frac{1}{\tilde{\omega}_x}-\frac{1}{\sqrt{n}})}\norm{ \qty(\openone\otimes\mathcal{B}_x) \ket{\psi}} \ \ \ [\text{using Eq.~(\ref{suppstaterof})}].
    \end{aligned}
\end{equation}
Employing Eq.~\ref{tildeomega},
\begin{equation}\label{1byomegab}
    \frac{1}{\tilde{\omega}_x}\approx \frac{1}{\sqrt{n}}\qty(1-\frac{\Delta_x}{2n}).
\end{equation}
Substituting $\frac{1}{\tilde{\omega}_x}$, Eq.~(\ref{obsrobustmid1}) gives
\begin{equation}\label{obsrobustmid2}
    \begin{aligned}
    \norm{\qty(\frac{\openone\otimes\mathcal{B}_x}{\tilde{\omega}_x}-A_x\otimes \openone)\ket{\tilde{\psi}}} &\leq \sqrt{n}\qty(\frac{1}{2\sqrt{n}}+1)C_n \sqrt{\delta}+ \frac{1}{2n}\abs{\Delta_x} \ \ \ [\text{since }\norm{\mathcal{B}_x \ket{\psi}}=\sqrt{n}] \\
    &\leq \qty(\frac{1}{2}+\sqrt{n})C_n \sqrt{\delta}+ \frac{1}{2^{n+2}n^2}\sqrt{\delta} \ \ \ \ [\text{from Eq.~(\ref{newD})}] \\
    &=Q_n \sqrt{\delta}, \ \ \ Q_n:= \qty{\qty(\frac{1}{2}+\sqrt{n})C_n+\frac{1}{2^{n+2}n^2}}\sim \order{n^{\frac{1}{4}}}.
    \end{aligned}
\end{equation}
Substituting Eq.~(\ref{obsrobustmid2}) back to Eq.~(\ref{arob}), we find
     \begin{equation}\label{obsrobustmid3}
  \norm{\qty(\tilde{A}_x-A_x)\otimes \openone \ket{\tilde{\psi}}} \leq \qty(F_n+Q_n) \sqrt{\delta} +\frac{1}{\tilde{\omega}_x}\norm{\openone\otimes \qty(\tilde{\mathcal{B}}_x-\mathcal{B}_x)\ket{\tilde{\psi}}}.
\end{equation}
Starting from the definition of $\tilde{M}_x\ket{\tilde{\psi}}=\tilde{A}_x \otimes \openone \ket{\tilde{\psi}}-\frac{1}{\tilde{\omega}_x}(\openone\otimes\tilde{\mathcal{B}}_x)\ket{\tilde{\psi}}$, we immediately express Bob’s scaled (unnormalised) physical observable in terms of Alice’s physical observable and the residual error operator, acting on the physical state
\begin{equation}\label{bobphysose}
    \frac{1}{\tilde{\omega}_x}(\openone\otimes\tilde{\mathcal{B}}_x)\ket{\tilde{\psi}}=\tilde{A}_x \otimes \openone \ket{\tilde{\psi}}-\tilde{M}_x\ket{\tilde{\psi}}.
\end{equation}
This equation clearly shows that Bob’s observable is fully determined by Alice’s action on the physical state, up to the correction $\tilde{M}_x$. In the ideal, noiseless case, $\tilde{M}_x\ket{\psi}=0$, Bob’s scaled observable is therefore entirely fixed by Alice’s observable acting on the ideal state $\ket{\psi}$. Consequently, Bob’s ideal scaled observables can be represented as a linear combination of Alice’s physical operators on the subspace defined by the ideal state $\ket{\psi}$ as
\begin{equation}\label{bobphytarsose}
  \frac{1}{\sqrt{n}}\qty( \openone \otimes \mathcal{B}_x\ket{\psi})= \tilde{A}_x \otimes \openone \ket{\psi}.
\end{equation}
Here, we deliberately use $\tilde{A}_x$ instead of $A_x$ to emphasise that, in the ideal limit, Bob’s observables are fully determined by Alice’s physical observable acting on the ideal subspace. Subtracting Eq.~(\ref{bobphysose})-(\ref{bobphytarsose}), we obtain
\begin{equation}\label{bobdiffphytar1}
\frac{1}{\tilde{\omega}_x}(\openone\otimes\tilde{\mathcal{B}}_x)\ket{\tilde{\psi}}-  \frac{1}{\sqrt{n}}\qty( \openone \otimes \mathcal{B}_x\ket{\psi})=\tilde{A}_x \otimes \openone \qty(\ket{\tilde{\psi}}-\ket{\psi}) - \tilde{M}_x\ket{\tilde{\psi}}.
\end{equation}
Considering the LHS of Eq.~(\ref{bobdiffphytar1}),
\begin{equation}\label{bobdiffphytar2}
    \begin{aligned}
   \frac{1}{\tilde{\omega}_x}(\openone\otimes\tilde{\mathcal{B}}_x)\ket{\tilde{\psi}}-  \frac{1}{\sqrt{n}}\qty( \openone \otimes \mathcal{B}_x\ket{\psi})&= \frac{1}{\tilde{\omega}_x}\qty{\openone\otimes\qty(\tilde{\mathcal{B}}_x-\mathcal{B}_x)}\ket{\tilde{\psi}}+   \frac{1}{\tilde{\omega}_x}(\openone\otimes\mathcal{B}_x)\ket{\tilde{\psi}} - \frac{1}{\sqrt{n}}(\openone\otimes\mathcal{B}_x)\ket{\psi}  \\
   &=\frac{1}{\tilde{\omega}_x}\qty{\openone\otimes\qty(\tilde{\mathcal{B}}_x-\mathcal{B}_x)}\ket{\tilde{\psi}} + \frac{1}{\sqrt{n}} \openone\otimes \mathcal{B}_x\qty(\ket{\tilde{\psi}}-\ket{\psi})-\frac{\Delta_x}{2n\sqrt{n}} \openone\otimes \mathcal{B}_x \ket{\tilde{\psi}}.
    \end{aligned}
\end{equation}
Thus, from Eqs.~(\ref{bobdiffphytar1}) and (\ref{bobdiffphytar2}), we get
\begin{equation}\label{bobdiffphytar3}
\frac{1}{\tilde{\omega}_x}\qty{\openone\otimes\qty(\tilde{\mathcal{B}}_x-\mathcal{B}_x)}\ket{\tilde{\psi}}= \tilde{A}_x \otimes \openone \qty(\ket{\tilde{\psi}}-\ket{\psi}) - \tilde{M}_x\ket{\tilde{\psi}}-   \frac{1}{\sqrt{n}} \openone\otimes \mathcal{B}_x\qty(\ket{\tilde{\psi}}-\ket{\psi})+\frac{\Delta_x}{2n\sqrt{n}} \openone\otimes \mathcal{B}_x \ket{\tilde{\psi}}.
\end{equation}
Taking norm and suitably applying triangle inequality to Eq.~(\ref{bobdiffphytar3}), we get
\begin{equation}\label{bobdiffphytar4}
\begin{aligned}
     \frac{1}{\tilde{\omega}_x}\norm{\openone\otimes\qty(\tilde{\mathcal{B}}_x-\mathcal{B}_x)\ket{\tilde{\psi}}} & \leq  \norm{\ket{\tilde{\psi}}-\ket{\psi}}+\norm{\tilde{M}_x\ket{\tilde{\psi}}}+\frac{1}{\sqrt{n}}\norm{\mathcal{B}_x}\cdot \norm{\ket{\tilde{\psi}}-\ket{\psi}}+\frac{1}{2n\sqrt{n}}\norm{\Delta_x}\cdot \norm{\mathcal{B}_x}   \\
     & \leq C_n \sqrt{\delta} + F_n \sqrt{\delta} + C_n \sqrt{\delta} + \frac{1}{2^{n+2}n^2}\sqrt{\delta} \ \ \ [\text{using $\norm{\mathcal{B}_x\ket{\psi}}=\omega_x=\sqrt{n}$}] \\
     &= H_n \sqrt{\delta}, \ \ \ \text{with } H_n:=2 C_n +F_n + \frac{1}{2^{n+2}n^2}.
\end{aligned}
\end{equation}
Finally, substituting Eq.~(\ref{bobdiffphytar4}) into Eq.~(\ref{obsrobustmid3}), we obtain the robust bound of Alice's observables
\begin{equation}\label{alicefinalrob}
  \norm{\qty(\tilde{A}_x-A_x)\otimes \openone \ket{\tilde{\psi}}} \leq D_n \sqrt{\delta} \ \ \ \text{with } D_n:= F_n+Q_n+H_n\sim \order{n^{\frac{1}{4}}}.
\end{equation}
Using the definitions $\mathcal{B}_x=\sum_{y=1}^n (-1)^{z^x_y} B_y$ and its physical counterpart $\tilde{\mathcal{B}}_x=\sum_{y=1}^n (-1)^{z^x_y} \tilde{B}_y$, we can invert this linear relation to express each of Bob’s individual physical observables in terms of the scaled observables $\tilde{\mathcal{B}}$. Applying the inverse transformation yields
\begin{equation}\label{bobdiffphytar3a}
  \openone\otimes\qty(\tilde{B}_y-B_y)\ket{\tilde{\psi}} = \openone\otimes \qty{\frac{1}{2^{n-1}} \sum_{x=1}^{2^{n-1}}(-1)^{z^x_y}\qty(\tilde{\mathcal{B}}_x-\mathcal{B}_x)}\ket{\tilde{\psi}}.
\end{equation}
This equation shows that the deviation of Bob’s individual observable from its ideal counterpart is determined by a linear combination of the deviations of the scaled observables. Taking the norm on both sides and applying the triangle inequality, we find
\begin{equation}\label{bobdiffphytar3b}
\norm{\openone\otimes\qty(\tilde{B}_y-B_y)\ket{\tilde{\psi}}} \leq \frac{1}{2^{n-1}} \sum_{x=1}^{2^{n-1}}\abs{(-1)^{z^x_y}} \cdot  \norm{\openone\otimes \qty(\tilde{\mathcal{B}}_x-\mathcal{B}_x)\ket{\tilde{\psi}} } \leq E_n \sqrt{\delta} \ \ [\text{using Eq.~(\ref{bobdiffphytar4})}]
\end{equation}
where $E_n:=\sqrt{n}H_n \sim \order{n^{\frac{1}{4}}}$.

\twocolumngrid
\bibliography{references} 
\end{document}